\documentclass[envcountsame]{llncs}

\usepackage{stmaryrd} 
\usepackage{amsfonts}
\usepackage{amssymb}
\usepackage{amsmath}
\usepackage{scalerel}
\usepackage{xspace}
\usepackage{paralist}
\usepackage{bm}
\usepackage{mfirstuc}
\usepackage{mdframed}
\usepackage{url}
\usepackage{tablefootnote}
\usepackage{scalerel}
\usepackage[colorlinks]{hyperref}

\newcommand{\interv}[2]{\{ #1,\dots,#2 \}}
\newcommand{\emp}{\mathtt{emp}}
\newcommand{\ar}[1]{\#(#1)}
\newcommand{\fv}[1]{\mathit{fv}(#1)}
\newcommand{\dom}[1]{\mathit{dom}(#1)}

\newcommand{\img}[1]{\mathit{img}(#1)}
\newcommand{\dunion}{\uplus}
\newcommand{\vars}{{\cal V}}
\newcommand{\size}[1]{\mathit{size}(#1)}
\newcommand{\card}[1]{\mathit{card}(#1)}
\newcommand{\len}[1]{\|#1\|}
\newcommand{\bigO}{\mathcal{O}}

\renewcommand{\vec}[1]{\mathbf{#1}}

\newcommand{\repl}[3]{#1\{ #2 \leftarrow #3 \}}
\newcommand{\replall}[4]{#1\{ #2 \leftarrow #3 \mid #4 \}}

\newcommand{\myset}[1]{\left\{#1\right\}}

\newcommand{\pcSID}{pc-SID\xspace}
\newcommand{\pcSIDs}{pc-SIDs\xspace}

\newcommand{\iseq}{\approx}

\newcommand{\Loc}{{\cal L}}

\newcommand{\astore}{\mathfrak{s}}
\newcommand{\aheap}{\mathfrak{h}}

\newcommand{\unfoldto}[1]{\Leftarrow_{#1}}
\newcommand{\preds}{{\cal P}_S}

\newcommand{\asid}{{\cal R}}
\newcommand{\rank}{\kappa}
\newcommand{\alloccompatible}{$\allocf$-compatible\xspace}

\newcommand{\locs}[1]{\mathit{loc}(#1)}

\newcommand{\bigAnd}{\scaleobj{1.5}{*}}

\newcommand{\modelst}{\models_{\theory}}

\newcommand{\modelsr}{\models_{\asid}}
\newcommand{\equivr}{\equiv_{\asid}}
\newcommand{\modelssid}[1]{\models_{#1}}

\newcommand{\allocf}{\mathit{alloc}}
\newcommand{\alloc}[1]{\allocf(#1)}
\newcommand{\vdashr}{\vdash_{\asid}}
\newcommand{\vdashsid}[1]{\vdash_{#1}}

\newcommand{\theory}{{\cal T}}
\newcommand{\tformula}{$\theory$-formula\xspace}
\newcommand{\tatom}{$\theory$-atom\xspace}
\newcommand{\slformula}{$SL$-formula\xspace}

\newcommand{\tpreds}{{\cal P}_{\cal T}}

\newcommand{\widt}[1]{\mathit{width}(#1)}

\newcommand{\aform}{\phi}
\newcommand{\aformB}{\psi}
\newcommand{\aformC}{\gamma}

\newcommand{\atform}{\chi}
\newcommand{\atformB}{\xi}
\newcommand{\anatom}{\alpha}

\newcommand{\aseq}{\Gamma}

\newcommand{\nil}{\mathrm{nil}}
\newcommand{\myfalse}{\mathtt{false}}
\begin{document}

\title{Two Results on Separation Logic 
With Theory Reasoning}
\author{Mnacho Echenim\inst{1} \and  Nicolas Peltier\inst{1}}

\institute{Univ. Grenoble Alpes, CNRS, LIG, F-38000 Grenoble France}
\authorrunning{M. Echenim and N. Peltier}

\titlerunning{A Proof Procedure For Separation Logic} 

%\authorrunning{Echenim et al.}

\maketitle

\newcommand{\myabstract}
{
Two results are presented concerning the entailment problem in Separation Logic with inductively defined predicate symbols and theory reasoning.
First, we show that the entailment problem is undecidable for rules satisfying the conditions given in \cite{IosifRogalewiczSimacek13}, if theory reasoning is considered. The result holds for a wide class of theories, even with a very low expressive power. For instance it applies to  the natural numbers with the successor function, or with the usual order.
Second, we show that every entailment problem can be reduced to an entailment problem containing no equality (neither in the formulas nor in the recursive rules defining the semantics of the predicate symbols).
}

\section{Introduction}

\newcommand{\lst}{\mathtt{ls}}
\newcommand{\slst}{\mathtt{ils}}
\newcommand{\alst}{\mathtt{als}}

\newcommand{\twoexptime}{$2$-$\mathsf{EXPTIME}$}

%Separation Logic (SL) \cite{IshtiaqOHearn01,Reynolds02},  is a well-established framework for 
%reasoning on programs manipulating pointer-based data structures. 
%It forms the basis of several industrial-scale static program
%analyzers
%\cite{DBLP:conf/nfm/CalcagnoDDGHLOP15,DBLP:conf/cav/BerdineCI11,DBLP:conf/cav/DudkaPV11}.
%The logic uses specific connectives to assert that formulas are satisfied on disjoint parts of the memory,
%which allows for more concise and more natural specifications.
In Separation Logic  (see, e.g., \cite{IshtiaqOHearn01,Reynolds02}), recursive data structures are usually specified using inductively defined predicates. The recursive rules defining the semantics of these predicates may be provided by the user. This specification
mechanism is similar to the definition of a recursive data type in an
imperative programming language.
For instance, a nonempty  list segment may be specified by using the inductive rules below, where the atom $x \mapsto (y)$ states that the memory location corresponding to $x$ is allocated and refers to $y$. %the pair (record) $(y,z)$ \mn{Maybe use $x\mapsto (y)$ to match example?}, where $y$ is the first element of the list and $z$ its tail. 
The symbol $*$ is a special logical connective denoting the disjoint composition of heaps.
\[ 
\lst(x,y) \Leftarrow x \mapsto (y) \qquad \lst(x,y) \Leftarrow \exists x'~.~ (x \mapsto (x') * \lst(x',y))
\]
Sorted lists with elements inside the interval $[u,v]$ may be specified as follows.
\[
\slst(x,y,u,v) \Leftarrow  (x \mapsto (y)  \wedge x \leq v \wedge x \geq u)\]
\[
\slst(x,y,u,v) \Leftarrow  \exists x'~.~ (x \mapsto (x') * \slst(x',y,x,v) \wedge x \leq v \wedge x \geq u)
\]
Many verification tasks can be reduced to an entailment problem in this logic.
 For example, to verify that some formula $\phi$ is a loop invariant, 
 we have to prove that  the weakest pre-condition of $\phi$ w.r.t.\ a finite sequence of 
 transformations is a logical consequence of $\phi$. Since techniques exist for computing automatically such pre-conditions, the problem may be reduced to an entailment problem between two SL formulas. 
Entailment problems  can also be used to 
 express typing properties.  
 For instance one may have to check that the entailment 
$\slst(x,y,u,v) \models \lst(x,y)$ holds, i.e., that a sorted list is a list, 
that 
$v \leq v' \wedge \slst(x,y,u,v) \models \slst(x,y,u,v')$ holds (if $v$ is an upper bound then 
any number greater than $v$ is also an upper bound) 
or that
$\slst(x,y,u,v) * \slst(y,z,v,v') \models \slst(x,y,u,v')$ (the composition of two sorted lists is a sorted list).
%Most available techniques for deciding entailment in this framework focus on particular ``hard-coded'' classes of structures, such as lists or trees.
%Investigating the entailment problem for SL formulas is thus of theoretical and practical interest.
%In practice, it is essential to offer as much flexibility as possible, and to
%handle a wide class of user-defined data structures (e.g., doubly linked lists, trees with child-parent links, trees with chained leaves etc.), possibly involving external theories, such as arithmetic.
In general, the entailment problem is undecidable 
%for formulas
%containing inductively defined predicates
\cite{DBLP:conf/atva/IosifRV14}, 
and a lot of effort has been devoted to identifying decidable fragments and devising proof procedures, see e.g., 
\cite{berdine-calcagno-ohearn04,CalcagnoYangOHearn01,cook-haase-ouaknine-parkinson-worell11,spen,DBLP:journals/fmsd/EneaLSV17,DemriGalmicheWendlingMery14}.
%A polynomial time algorithm for deciding entailment problems for SL formulas over lists was defined in \cite{}.
In particular, a very general class of  decidable  entailment problems is described in
\cite{IosifRogalewiczSimacek13}.
% It is based on the decidability of the satisfiability
%problem for monadic second order logic over graphs of a bounded treewidth, for formulas involving no theory other than equality.
This fragment does not allow for any theory predicate other than equality and is defined by restricting the form of the inductive rules, which must fulfill $3$ conditions, formally defined below: the {\em progress} condition (every rule allocates a single memory location), 
the {\em connectivity} condition (the set of allocated locations has a tree-shaped structure)
and the {\em establishment} condition (every existentially quantified variable is eventually allocated). 
% (for instance, the rule $\lst(x,y) \Leftarrow \exists u\exists x'~.~ (x \mapsto (u,x') * \lst(x',y))$ is progressing and connected but not established, since $u$ is not allocated).
%Initially, no upper bound better than
%elementary recursive was known to exist. 
More recently, a
\twoexptime\ algorithm was proposed for such entailments
\cite{PMZ20}, and we showed  
in \cite{DBLP:conf/lpar/EchenimIP20}
 that this bound is 
tight. % and in \cite{EIP21a} we devised a new algorithm, handling more general  classes of inductive definitions.
%The algorithms in \cite{PMZ20,DBLP:conf/lpar/EchenimIP20} work by computing some abstract representation of the set of models of SL formulas.
%To test that  $\phi \models \psi$ holds, one must check that every structure satisfying $\phi$ also satisfies $\psi$. 
%To this aim, the algorithms 
%compute a representation of the set of structures satisfying $\phi$ bottom-up,
%starting from atoms, using abstract operators
%to compose disjoint structures  and to
%unfold the inductive rules associated with the predicates. 
%The
%abstraction is precise enough to allow checking that all the models of
%the $\phi$ are also models of $\psi$, and
%also general enough to ensure termination of the entailment checking
%algorithm. %\mycomment[np]{addition} 
%Other approaches have been proposed to check entailments in various fragments, see e.g., \cite{cook-haase-ouaknine-parkinson-worell11,DBLP:journals/fmsd/EneaLSV17,DBLP:conf/atva/IosifRV14}.
% In particular, a sound and complete proof procedure is given in \cite{DBLP:conf/aplas/TatsutaNK19} for inductive rules satisfying conditions that are strictly more restrictive than those in \cite{IosifRogalewiczSimacek13}.
%In \cite{DBLP:journals/logcom/GalmicheM21} a labeled proof systems is presented for separation logic formulas
%handling arbitrary inductive definitions and all connectives (including negation and separated implication).
%Due to the expressive power of the considered language, this proof system is of course not terminating or complete in general.
To tackle entailments such as those given above, one must be able to combine spatial reasoning with theory reasoning, and
%
%is of uttermost importance 
%for applications, as reasoning on data structures without taking into account the properties of the data stored in these structures has a  limited scope, and 
the combination of SL with data constraints has been considered by several authors (see, e.g., \cite{DBLP:conf/cav/PiskacWZ13,DBLP:conf/pldi/Qiu0SM13,DBLP:conf/aplas/PerezR13,DBLP:conf/cade/XuCW17,DBLP:conf/vmcai/Le21}). 
%For instance, in order to prove the validity of  the implications above involving sorted lists, one has to take into account
%the properties of the order $\leq$.
It is therefore natural to ask whether the above decidability result extends to the case where 
theory reasoning is considered.
In the present paper, we show that this is not the case, even for very simple theories.
More precisely, we establish two new results.
First, we show that the entailment problem is undecidable for rules satisfying the above conditions if theory reasoning is allowed (Theorem \ref{theo:undec}). The result holds for a very wide class of theories, 
even for theories with a very low expressive power. For instance, it holds
for the natural numbers with only the successor function,
or with only the predicate $\leq$ (interpreted as usual).
Second, we show that 
every entailment can be reduced to an entailment not containing equality (Theorem \ref{theo:elimeq}). The intuition is that  all the equality and disequality constraints
can be encoded in the formulas describing the shape of the data structures.
The transformation increases the number of rules exponentially but it increases the size of the rules
only polynomially (hence it preserves the complexity results in \cite{EIP21a}). 
This result shows that the addition of the equality predicate does not increase the expressive power. It may be useful 
to facilitate the definition of proof procedures for the considered fragment.

\section{Preliminaries}

\newcommand{\tmodel}{\modelsr_{\cal T}}
\newcommand{\hpredicate}{spatial predicate\xspace}
\newcommand{\tpredicate}{$\theory$-predicate\xspace}

In this section, we define the syntax and semantics of the fragment of separation logic that is considered in the paper (see for instance \cite{DBLP:conf/csl/OHearnRY01,Reynolds02,IosifRogalewiczSimacek13} for more details).

\paragraph*{Syntax.}

Let $\vars$ be a countably infinite set of {\em variables}.
We consider a set $\tpreds$ of %(possibly countably infinite)  
 {\em {\tpredicate}s} (or {\em theory predicates}, denoting relations in an underlying theory of locations) and 
a set  $\preds$ of {\em {\hpredicate}s}, disjoint from $\tpreds$. Each symbol $p\in \tpreds \cup \preds$ is  associated with a unique arity $\ar{p}$.
We assume that $\tpreds$ contains in particular two binary symbols $\iseq$ and $\not \iseq$ and a nullary symbol $\myfalse$.

\begin{definition}
Let $\rank$ be some fixed 
natural number.
The set of {\em {\slformula}s} (or simply formulas) $\aform$ is inductively defined as follows:
\[\aform := \emp \; \| \; x \mapsto (y_1,\dots,y_\rank) \; \| \; \aform_1 \vee \aform_2 \; \| \;  \aform_1 * \aform_2 \| \; 
p(x_1,\dots,x_{\ar{p}}) \; \| \; \exists x. ~ \aform_1   \]
where  $\aform_1,\aform_2$ are {\slformula}s, $p\in \tpreds \cup \preds$ and $x,x_1,\dots,x_{\ar{p}}, y_1,\dots,y_{\rank}$ are variables.
\end{definition}

% A {\em spatial atom} is  a formula that is either 
A formula of the form $x \mapsto (y_1,\dots,y_\rank)$ is called a {\em points-to atom}, and a formula 
% or 
$p(x_1,\dots,x_{\ar{p}})$ with $p\in \preds$ is called a {\em predicate atom}. A {\em spatial atom} is either a points-to atom or a predicate atom.
A {\em \tatom} is a formula of the form $p(x_1,\dots,x_{\ar{p}})$ with $p\in \tpreds$.
An {\em atom} is either a spatial atom or a \tatom.
A {\em \tformula} is either $\emp$ or a separating conjunction of {\tatom}s.
A formula of the form $\exists x_1.\dots.\exists x_n.~\aform$ (with $n \geq 0$) is denoted by 
$\exists \vec{x}.~\aform$, where $\vec{x} = (x_1,\dots,x_n)$.
A formula is {\em predicate-free} (resp.\ {\em disjunction-free}, resp. {\em quantifier-free}) if it contains no predicate symbol in $\preds$ (resp.\ no occurrence of $\vee$, resp.\ of $\exists$).
It is in {\em prenex form} if it is of the form $\exists \vec{x}. \aform$, where $\aform$ is quantifier-free and $\vec{x}$ is a possibly empty vector of variables.
A {\em symbolic heap} is a prenex disjunction-free formula, i.e., a formula of the form $\exists \vec{x}. \aform$, where $\aform$ is  a separating conjunction of atoms.

Let $\fv{\aform}$ be the set of variables freely occurring in $\aform$. % (i.e., occurring in $\aform$ but not within the scope of any existential quantifier).  
  We assume (using $\alpha$-renaming if needed) that all existential variables are distinct from free variables and that
 distinct occurrences of quantifiers bind distinct variables.
 A {\em substitution} $\sigma$ is a function mapping variables to variables. The {\em domain} $\dom{\sigma}$ of a substitution $\sigma$ is the set of variables $x$ 
such that $\sigma(x) \not = x$, 
and we let $\img{\sigma} = \sigma(\dom{\sigma})$. For any expression (variable, tuple of variables or formula) $e$, we denote by $e\sigma$ the expression obtained from $e$ by replacing 
every free occurrence of a variable $x$ by $\sigma(x)$ and by $\replall{}{x_i}{y_i}{1 \leq i \leq n}$ (where the $x_1,\dots,x_n$ are pairwise distinct) the substitution 
such that $\sigma(x_i) = y_i$ and $\dom{\sigma} \subseteq \{ x_1,\dots,x_n \}$.
 For all sets $E$, 
 $\card{E}$ 
 is the cardinality of $E$. 
 For all sequences or words $w$, 
 $\len{w}$ denotes the length of $w$.
 We sometimes identify vectors with sets, if the order is unimportant, e.g., we may write $\vec{x} \setminus \vec{y}$ to denote the vector formed by the components of $\vec{x}$ that do not occur in $\vec{y}$.

  We assume that the symbols in $\preds \cup \tpreds \cup \vars$ are words
  %\footnote{Because we will consider transformations introducing an unbounded number of new predicate symbols, we cannot assume that the predicate atoms have a  constant size.} 
  over a finite alphabet of some constant size, strictly greater than $1$. For any expression $e$, we denote by $\size{e}$ 
  the size of $e$, i.e., the number of occurrences of 
symbols\footnote{Each symbol $s$ in $\preds \cup \tpreds \cup \vars$ is counted with a weight equal to its length $\len{s}$, and all the logical symbols have weight $1$.} in $e$.
 We define the {\em width} of a formula as follows: 
 %\mycomment[me]{remove the 'if $\aform$ is disjunction-free' in last equation?}\mycomment[np]{if $\aform$ is not disjunction-free then we may have $\widt{\aform}\not = \size{\aform}$}
{\small
\[
\begin{tabular}{llllll}
$\widt{\aform_1 \vee \aform_2}$ & $ = $ & $\max(\widt{\aform_1},\widt{\aform_2})$ & \qquad  
   $\widt{\exists x. \aform}$ & $=$ & $\widt{\aform} + \size{\exists x}$ \\
$\widt{\aform_1 * \aform_2}$ & $ = $ & $\widt{\aform_1} + \widt{\aform_2} + 1$ \\
$\widt{\aform}$ & $ = $ & $\size{\aform}$ \quad if $\aform$ is an atom \\
\end{tabular}
\]
}

Note that $\widt{\aform}$ coincides with $\size{\aform}$ if $\aform$ is disjunction-free.

\paragraph*{Inductive Rules.}

\newcommand{\dependson}[1]{\geq_{#1}}

The semantics of the predicates in $\preds$ is given  %defined \Mnacho{given? provided?}
by user-defined inductive rules. To ensure decidability in the case where the theory only contains the equality predicate, these rules must satisfy some additional conditions (defined in \cite{IosifRogalewiczSimacek13}):
\begin{definition}
\label{def:sid}
A (progressing and connected) set of inductive rules (\pcSID) $\asid$
is a finite set of rules of the form
\( p(x_1,\dots,x_n) \Leftarrow \exists \vec{u}. ~ x_1 \mapsto (y_1,\dots,y_\rank) 
* \aform\)
	where $\fv{x_1 \mapsto (y_1,\dots,y_\rank) 
		* \aform} \subseteq \{x_1, \ldots, x_n\} \cup \vec{u}$,
		$\aform$ is a possibly empty separating conjunction of predicate atoms and {\tformula}s, and for every predicate atom $q(z_1,\dots,z_{\ar{q}})$ 
		%\mycomment[np]{deleted: (with $q\in \preds$) redundant if ``predicate" is added} 
		occurring in $\aform$, we have $z_1 \in \{ y_1,\dots,y_\rank\}$.
We let $\size{p(\vec{x}) \Leftarrow \aform} = \size{p(\vec{x})} + \size{\aform}$, 
$\size{\asid} = \Sigma_{\rho\in \asid} \size{\rho}$
and $\widt{\asid} = \max_{\rho \in \asid} \size{\rho}$.
\end{definition}

In the following, $\asid$ always denotes a \pcSID. 
Note that the right-hand side of  every inductive rule contains exactly one 
points-to atom, the left-hand side of which is the first argument $x_1$ of the predicate symbol (this condition is referred to as the {\em progress} condition), and
that this points-to atom contains the first argument of every predicate atom on the right-hand side of the rule (the {\em connectivity} condition).

%\mycomment[np]{addition:}

%\begin{proposition}
%For every SID $\asid$, $\card{\asid} = \bigO(2^{\widt{\asid} \times \ln(\widt{\asid})})$ (up to a renaming of variables).
%\end{proposition}

\begin{definition}
\label{def:unfold}
We write $p(x_1,\dots,x_{\ar{p}}) \unfoldto{\asid} \aform$
if $\asid$ contains a rule (up to $\alpha$-renaming) %of the form 
$p(y_1,\dots,y_{\ar{p}}) \Leftarrow \aformB$, where $x_1,\dots,x_{\ar{p}}$ are not bound in $\aformB$, and
$\aform = \replall{\aformB}{y_i}{x_i}{i \in \interv{1}{{\ar{p}}}}$.

The relation $\unfoldto{\asid}$ is extended to all formulas as follows:
$\aform \unfoldto{\asid} \aform'$ if one of the following conditions holds:
(i) $\aform = \aform_1 \bullet \aform_2$ (modulo AC, with $\bullet \in \{ *, \vee \}$), $\aform_1 \unfoldto{\asid} \aform_1'$, no free or existential variable  in $\aform_2$ is bound in $\aform_1'$
and 
$\aform' = \aform_1' \bullet \aform_2$; 
%(up to a transform into prenex form); 
or (ii) $\aform = \exists x.~ \aformB$, $\aformB \unfoldto{\asid} \aformB'$, $x$ is not bound in $\aformB'$
and 
$\aform' = \exists x.~ \aformB'$.
We denote by $\unfoldto{\asid}^+$ the transitive closure of $\unfoldto{\asid}$, and by $\unfoldto{\asid}^*$ its reflexive and transitive closure.
A formula $\aformB$ such that $\aform \unfoldto{\asid}^* \aformB$ is called an {\em $\asid$-unfolding} of 
$\aform$.
We denote by $\dependson{\asid}$ the least transitive and reflexive binary relation on $\preds$
such that $p \dependson{\asid} q$ holds if 
$\asid$ contains a rule of the form $p(y_1,\dots,y_{\ar{p}}) \Leftarrow \aformB$, where $q$ occurs in $\aformB$. If $\aform$ is a formula, we write $\aform \dependson{\asid} q$ if $p \dependson{\asid} q$ for some $p \in \preds$ occurring in $\aform$.
\end{definition}

\paragraph*{Semantics.}

\begin{definition}
Let $\Loc$ be a countably infinite set of %so-called \Mnacho{if there is a space issue,  'so-called' can be removed} 
{\em locations}.
An {\em SL-structure} is a pair $(\astore,\aheap)$ where 
$\astore$ is a {\em store}, i.e.\ a total function from $\vars$ to $\Loc$, and 
$\aheap$ is a {\em heap}, i.e.\ a partial finite function from $\Loc$ to $\Loc^\rank$ (written as a relation: $\aheap(\ell) = (\ell_1,\dots,\ell_\rank)$ iff $(\ell,\ell_1,\dots,\ell_\rank) \in \aheap$).
The {\em size} of a structure $(\astore,\aheap)$ is the cardinality  of $\dom{\aheap}$.
\end{definition}

For every heap $\aheap$, we define: $\locs{\aheap} = \{ \ell_i \mid (\ell_0,\dots,\ell_\rank) \in \aheap, i = 0,\dots,\rank \}$.
A location $\ell$ (resp.\ a variable $x$) is {\em allocated} in a heap $\aheap$ (resp.\ in a structure ($\astore,\aheap$)) if $\ell \in \dom{\aheap}$ (resp.\ $\astore(x)\in \dom{\aheap}$).
Two heaps $\aheap_1,\aheap_2$ are {\em disjoint} if 
$\dom{\aheap_1} \cap \dom{\aheap_2} = \emptyset$,
in this case $\aheap_1 \dunion \aheap_2$ denotes the 
union of $\aheap_1$ and $\aheap_2$.

%\mycomment[np]{removed reference to $\forall$ in case we get another autistic reviewer}

 Let $\modelst$ be a satisfiability relation between stores and {\tformula}s, satisfying the following properties:
  $\astore \modelst x \iseq y$ (resp.\ $\astore \modelst x \not \iseq y$) iff $\astore(x) = \astore(y)$ (resp.\ $\astore(x) \not = \astore(y)$), $\astore \not \modelst \myfalse$ and  
 %\mycomment[np]{added:} 
 $\astore \modelst \atform *\atformB$ iff $\astore\modelst \atform$ and $\astore \modelst \atformB$. 
%\mycomment[np]{added:}  In all examples, $\Loc$ is the set of integers, and the {\tformula}s are arithmetic formulas, interpreted as usual.
% This relation is extended to quantifications of {\tformula}s in the usual way, 
% i.e., we write 
% $\astore \modelst \forall \vec{x}. \atform$ (resp.\ $\astore \modelst \exists \vec{x}. \atform$) if  
% $\astore' \modelst \atform$ holds for every (resp.\ for some) $\astore'$ coinciding with $\astore$ on 
% all variables not occurring in $\vec{x}$.
  For all {\tformula}s $\atform, \atformB$, we write 
$\atform \modelst \atformB$ if 
$\astore \modelst \atform \implies \astore \modelst \atformB$ holds for all stores $\astore$. 

  \begin{definition}
 \label{def:semantics}
Given formula $\aform$, a \pcSID $\asid$ and a structure $(\astore,\aheap)$,
we write $(\astore,\aheap) \modelsr \aform$ and say that $(\astore,\aheap)$ is an {\em $\asid$-model} (or simply a model if $\asid$ is clear from the context) of $\aform$ if one of the following conditions holds.
\begin{compactitem}
\item{$\aform = x \mapsto (y_1,\dots,y_\rank)$ and
$\aheap = \{ (\astore(x),\astore(y_1),\dots,\astore(y_\rank)) \}$.}
\item{$\aform$ is a \tformula, $\aheap = \emptyset$ and $\astore \modelst \aform$.}
\item{$\aform = \aform_1 \vee \aform_2$ and 
$(\astore,\aheap) \modelsr \aform_i$, for some $i  = 1,2$.}
\item{$\aform = \aform_1 * \aform_2$ and there exist disjoint heaps
$\aheap_1,\aheap_2$ such that 
$\aheap = \aheap_1 \dunion \aheap_2$ and 
$(\astore,\aheap_i) \modelsr \aform_i$, for all $i  = 1,2$.}

\item{$\aform = \exists x. ~ \aform$ and 
$(\astore',\aheap) \modelsr \aform$, for some store $\astore'$ coinciding with 
$\astore$ on all variables distinct from $x$.}
\item{
$\aform = p(x_1,\dots,x_{\ar{p}})$, $p \in \preds$ and  $(\astore,\aheap) \modelsr \aformB$ for some $\aformB$ such that
$\aform \unfoldto{\asid} \aformB$.}
\end{compactitem}
If $\aseq$ is a sequence of formulas, then we write $(\astore,\aheap) \modelsr \aseq$ if $(\astore,\aheap)$ satisfies at least one formula in $\aseq$.
\end{definition}
We emphasize that a \tformula is satisfied %fulfilled \Mnacho{satisfied?} 
only 
in structures with empty heaps. This convention is used to simplify notations, because it avoids 
having to consider both standard and separating conjunctions.
Note that Definition \ref{def:semantics} is well-founded because of the progress condition: the size of $\aheap$ decreases at each recursive call of a predicate atom.
We write $\aform \modelsr \aformB$ if every $\asid$-model of $\aform$ is an $\asid$-model of $\aformB$
and  $\aform \equivr \aformB$ if $\aform \modelsr \aformB$ and $\aformB \modelsr \aform$.

Every formula can be transformed into prenex form using the well-known equivalences: $(\exists x. \aform) \bullet \aformB \equiv \exists x. (\aform \bullet \aformB)$, for all $\bullet \in \{ \vee, * \}$, where $x\not \in\fv{\aformB}$.

\paragraph*{Establishment.}

The notion of establishment \cite{IosifRogalewiczSimacek13} is defined as follows:
 
\begin{definition}
A \pcSID is {\em established}
if
for every atom $\anatom$,  every predicate-free  formula 
$\exists \vec{x}. \aform$ such that
$\anatom \unfoldto{\asid}^* \exists \vec{x}. \aform$ (up to a transformation into prenex form) and $\aform$ is quantifier-free, and  every $x\in \vec{x}$, $\aform$ is of the form $x' \mapsto (y_1,\dots,y_\rank) * \atform * \aformB$, where $\atform$ is a separating conjunction of equations (possibly $\emp$) such that
 $\atform \modelst x \iseq x'$. 
\end{definition}

In the remainder of the paper, we assume that every considered \pcSID is established.

\paragraph*{Sequents.}

We consider sequents denoting entailment problems and defined as follows:

\begin{definition}
\label{def:sequent}
A {\em sequent} is an expression of the form $\aform_0 \vdashr \aform_1,\dots,\aform_n$, where $\asid$ 
is a \pcSID and $\aform_0,\dots,\aform_n$ are formulas.
A sequent is {\em disjunction-free} if $\aform_0,\dots,\aform_n$ are disjunction-free, and {\em established} if $\asid$ is established.
We define: 
{\small
\begin{eqnarray*}
	\size{\aform_0 \vdashr \aform_1,\dots,\aform_n} &=& \Sigma_{i=0}^n \size{\aform_i} + \size{\asid}, \qquad \fv{\aform_1,\dots,\aform_n} = \bigcup_{i=0}^n \fv{\aform_i},\\
	\widt{\aform_0 \vdashr \aform_1,\dots,\aform_n} &=& \max \{  \widt{\aform_i}, \widt{\asid}, \card{ \bigcup_{i=0}^n \fv{\aform_i}} \mid 0 \leq i \leq n \}.
\end{eqnarray*}}
%$\size{\aform_0 \vdashr \aform_1,\dots,\aform_n} = \Sigma_{i=0}^n \size{\aform_i} + \size{\asid}$, 
%and 
%$\widt{\aform_0 \vdashr \aform_1,\dots,\aform_n} = \max \{  \widt{\aform_i}, \widt{\asid} \mid 0 \leq i \leq n \}$.
\end{definition}
%The restriction to points-free formulas in the right-hand side is for technical convenience only. It is not restrictive,
%since a points-to atom $x \mapsto (y_1,\dots,y_\rank)$ may be replaced by an atom $p(x,y_1,\dots,y_\rank)$ where $p$ is a fresh predicate symbol associated with a unique inductive rule
%$p(x,y_1,\dots,y_\rank) \Leftarrow x \mapsto (y_1,\dots,y_\rank)$.
%We also assume that $\aform_1,\dots,\aform_n$ contain no predicate atom outside of $\swand$, using the equivalence
%$p(\vec{x}) \equiv \MW{\emp}{p(\vec{x})}$.

\begin{definition}
\label{def:counter_model}
 A structure $(\astore,\aheap)$ is a {\em countermodel} of 
a sequent  $\aform \vdashr \aseq$ iff $\astore$ is injective,  $(\astore,\aheap) \modelsr \aform$
 and $(\astore,\aheap) \not \modelsr \aseq$.
A sequent is {\em valid} if it has no countermodel.
Two sequents are {\em equivalent} if they are both valid or both non-valid\footnote{Hence two non-valid sequents with different countermodels are equivalent.}.
\end{definition}
The restriction to injective countermodels is for technical convenience only and does not entail any loss of generality.

 \section{An Undecidability Result}

\label{sect:undec}

This section contains the main result of the paper. It shows that no terminating procedure for checking the validity of
(established) sequents possibly exists, even for theories with a very low expressive power.

\newcommand{\succp}{S}
\newcommand{\succs}{{\mathfrak S}}
\newcommand{\nsuccp}{\overline{S}}

\newcommand{\lsol}{W}
\renewcommand{\nil}{\mathrm{\bot}}
\newcommand{\startp}{\mathtt{B}}
\newcommand{\stopp}[1]{\mathtt{E}_{#1}}
\newcommand{\pairs}{\mathtt{P}}
\newcommand{\nextp}[1]{\rightarrow^#1}
\newcommand{\myrel}{\triangleleft}
%Mnacho: suggestion for \myrel
%\newcommand{\myrel}{\circeq}

\newcommand{\aloc}{\alpha}

%\mycomment[np]{change in the formulation below for compactness}

\begin{theorem}
\label{theo:undec}
The validity problem is undecidable 
 for established sequents $\aform \vdashr \aformB$ if
 $\tpreds$ contains  predicates $\succp$ and $\nsuccp$, where:
\begin{compactitem}
\item{
$\succp$ is interpreted as a
relation $\succs$ satisfying the following property: there exists a set of pairwise distinct locations 
$\{ \aloc_i, \aloc_i', \aloc_i'' \mid i \in {\Bbb N} \}$ such that for all $i \in {\Bbb N}$, $(\aloc_i,\aloc_i') \in \succs$, $(\aloc_i'',\aloc_i') \not \in \succs$, and 
%\mycomment[np]{small correction below ($\iseq$ replaced by $=$)}
for all locations $\ell \in \{ \aloc_j,\aloc_j', \aloc_j'' \mid j \in {\Bbb N} \}$ if $\aloc_i \not = \ell$, $(\aloc_i,\ell) \in \succs$ and $(\aloc_{i}'',\ell) \not \in \succs$, then $\ell = \aloc_{i}'$;}
\item{%\mycomment[np]{changed:} 
	$\nsuccp(x,y)$ and $\neg \succp(x,y)$ are interpreted equivalently when $x$ and $y$ refer to distinct locations.}
\end{compactitem}
%Then the validity problem is undecidable for established sequents $\aform \vdashr \aformB$.
\end{theorem}
For instance, the hypotheses of Theorem \ref{theo:undec} are trivially satisfied on the natural numbers
if $\succs$ is the successor function
%is an injective  function  (e.g., the successor relation on natural numbers, or the function $x \mapsto -x$ on integers) and $\nsuccp(x,y) \equiv \neg \succp(x,y)$, 
or if $\succs$ is the usual order $\leq$, and $\nsuccp$ is $\geq$ (with $\aloc_i = 3.i$, $\aloc_i' = 3.i+1$, $\aloc_i'' = 3.i+2$ in both cases, since $\aloc_i+1 \iseq \ell \Rightarrow \aloc_i' \iseq \ell$ and $\aloc_i \leq \ell \wedge \aloc_{i}'' > \ell \wedge \ell \not = \aloc_i \Rightarrow \aloc_{i}' \iseq \ell$. ). 
More generally, the conditions hold if the domain is infinite and $\succs$ is any injective  function $f$ such that $f(x) \not = x$.
In this case, the sequences $\aloc_i,\aloc_i',\aloc_i''$ may be constructed inductively: for every $i \in {\Bbb N}$,  $\aloc_i$ is any element $e$ such that both $e$ and $f(e)$ do not occur in $\{ \aloc_j,\aloc_j',\aloc_j'' \mid j < i \}$, 
$\aloc_i'$ is $f(\aloc_i)$ and 
$\aloc_i''$  is any element not occurring in $\{ \aloc_j,\aloc_j',\aloc_j'' \mid j < i \}  \cup \{ \aloc_i,\aloc_i' \}$. %\mycomment[np]{was wrong: any non well-founded (strict or non strict) order}
%\mycomment[np]{added:} 
Note that in this case the locations $\aloc_i''$ are actually %\mycomment[me]{zap? useless and} fine
irrelevant, but
these locations play an essential r\^ole in the undecidability proof 
%\mycomment[me]{in the incompleteness proof}\mycomment[np]{undecidability instead of incompleteness?} 
when $\succs$ is $\leq$.
The remainder of the section is devoted the proof of Theorem \ref{theo:undec}.

\begin{proof}
%\optionalProof{Theorem \ref{theo:undec}}
%\optionalProof{Theorem \ref{theo:undec}}
%\mycomment[np]{cosmetic changes}\mycomment[me]{Many modifs, previous version commented, to recheck!}
	The proof goes by a reduction from the Post Correspondence Problem (PCP).
	We recall that the PCP consists in determining, given two sequences of words
	$u = (u_1,\dots,u_n)$ and $v = (v_1,\dots,v_n)$, whether there exists a nonempty sequence $(i_1,\dots,i_k)\in \interv{1}{n}^k$ such that
	$u_{i_1}.\dots,u_{i_k} = v_{i_1}.\dots,v_{i_k}$. It is well-known that this problem is undecidable.
	We assume, w.l.o.g., that
	 $\len{u_i} > 1$
	  and $\len{v_i} > 1$ 
	  for all $i \in \interv{1}{n}$.
	A word $w$ such that $w = u_{i_1}.\dots,u_{i_k} = v_{i_1}.\dots,v_{i_k}$  is called a {\em witness}.
	Positions  inside words of the sequences $(u_1,\dots,u_n)$ and $(v_1,\dots,v_n)$ will be denoted by pairs $(i,j)$, 
	encoding the $j$-nth character of the words $u_i$ or $v_i$. More formally, if $p = (i,j)$, and $w \in \{u,v\}$, then  we denote by $w(p)$ the $j$-th symbol of the word $w_i$, provided this symbol is defined. 	 We write $p \myrel q$ if
	both $u(p)$ and $v(q)$ are defined and $u(p) = v(q)$. 
%	\mycomment[np]{I use now $\myrel$ instead of $\bowtie$ because the relation is not symmetric} \mycomment[me]{to discuss, why isn't it symmetric?}\mycomment[np]{take eg 
%	$u = (ac), v = (ba)$, $p = (1,1)$, $q = (1,2)$. We have 
%	$u(p) = a = v(q)$ and $u(q) = c \not = b = v(p)$. Of course the relation is symmetric if we also switch $u$ and $v$, but here $\myrel$ is fixed.}
		
	Let $m = \max \{ \len{u_i}, \len{v_i} \mid i \in \interv{1}{n} \}$. We denote by 
	$\pairs$ the {set of} pairs of the form $(i,j)$ with $i\in \interv{1}{n}$ and $j \in \interv{1}{m}$, and by 
	$\startp$ the {set of} pairs of the form $(i,1)$. For $w\in \myset{u,v}$, we denote by $\stopp{w}$ the set of pairs of the form $(i,\len{w_i})$, where $i\in \interv{1}{n}$, and we write $(i,j) \nextp{w} (i',j')$ 
	either $i' = i$, $j< \len{w_i}$ and $j' = j+1$, %\mycomment[me]{instead of $i' = i$, $j' = j+1$ and $j' \leq \len{w_i}$} \mycomment[np]{yes}
	or $j = \len{w_i}$ and $j' = 1$. 
%	\mycomment[me]{and $i' = i+1$?}\mycomment[np]{no. $i$
%	is the index of the word in the initial sequence not in the solution. There is no reason that $i' = i+1$ (the next word in the solution is not necessarily the next one in the initial sequence). 
%	At this point we just generate structures encoding all 
%	the words obtained by concatening words in both sequences $u$ and $v$ in an arbitrary way thus $i'$ is arbitrary. Added:} 
Note that $i'$ is arbitrary in the latter case (intuitively 
$(i,j) \nextp{w} (i',j')$  states that the character corresponding to the position $(i,j)$ may be followed  in a witness by 
the character at position $(i',j')$).
	Let $\vec{v}$ be a vector of variables, where all elements $p \in \pairs$  
	are associated with pairwise distinct variables in $\vec{v}$. To simplify notations, we will  also denote by $p$ the variable associated with $p$. We assume the vector 
	$\vec{v}$ also contains a special variable $\nil$, distinct from the variables $p \in \pairs$.
	We construct a representation of potential witnesses as heaps.
	The encoding is given for $\rank = 6$, although in principle this encoding could be defined with $\rank = 2$
	by encoding tuples as binary trees.
	Witnesses are encoded by linked lists, with links on the last argument, 
	%\mycomment[np]{addition} 
	and starting with a dummy element $(\nil,\dots,\nil)$. Except for the first dummy element, each location 
	%\mycomment[me]{element?}\mycomment[np]{used location to avoid repetition of ``element'' feel free to change} 
	in the list 
	refers to two locations associated with pairs $p,q \in  \pairs$ denoting positions inside the two sequences 
	$u_1,\dots,u_n$ and $v_1,\dots,v_n$ respectively, and to three additional allocated locations the r\^ole of which will be detailed below.  
	%\mycomment[me]{in first rule, replace ``$p = (i,1)$ and $i \in \interv{1}{n}$''  by $p \in \startp$?}\mycomment[np]{yes}
	
	%\mycomment[np]{some corrections in rule 3 below (previous version was wrong)}

	{\small
		\[
		\begin{tabular}{rclr}
			$\lsol(x,\vec{v})$ & $\Leftarrow$ & $\exists x'. ~ x \mapsto (\nil,\nil,\nil,\nil,\nil,x') 
			* \lsol_{p,p}(x', \vec{v})$ \\
			& & \qquad where $p \in \startp$  \\
			$\lsol_{p,q}(x,\vec{v})$ &  $\Leftarrow$ & $\exists x',y,z,u. ~ x \mapsto (p,q,y,z,u,x') * \lsol_{p',q'}(x',\vec{v}) * P(y,\nil) * P(z,\nil)$ \\
			& & \quad  $*\ P(u,\nil)$ \qquad where $p\myrel q$, $p \nextp{u} p'$ and $q \nextp{v} q'$\\
			
			$\lsol_{p,q}(x,\vec{v})$ &  $\Leftarrow$ & $\exists y,z,u. ~ x \mapsto (p,q,y,z,u,\nil) * P(y,\nil) * P(z,\nil)$ \\
			& & $*\ P(u,\nil)$ \qquad where $p = (i,\len{u_i})$, $q = (i,\len{v_i})$, and $p \myrel q$ \\
			$P(x,y)$ & $\Leftarrow$ & $ x \mapsto (y,y,y,y,y,y)$ \\
		\end{tabular}
				\]
	}

	%\mycomment[np]{typos fixed below} 

	\newcommand{\succpbFirstCell}{A}
	\newcommand{\succpbInnerCell}{B}
	\newcommand{\pbSeq}{C}
	\newcommand{\setof}[2]{\left\{#1\,\|\:#2\right\}}
	\newcommand{\ellg}[2]{\ell^{#1}_{#2}}
	%\mycomment[me]{Tentative rewriting, to be validated before I carry on}
	
	By construction, the structures that validate $\lsol(x,\vec{v})$ are of the form $(\astore, \aheap)$, where the store $\astore$ verifies:
	\begin{compactitem}
		\item $\astore(x) = \ell$ and $\astore(\nil) = \ell'$;
		\item for all $i = 1,\dots,m'$, $\astore(p_i) = \ellg{p}{i}$ and $\astore(q_i) = \ellg{q}{i}$, where $p_i,q_i \in \pairs$ are such that $p_i \myrel q_i$,  $p_i \nextp{u} p_{i+1}$ and $q_i \nextp{v} q_{i+1}$,
	\end{compactitem} 
	and the heap $\aheap$ is of the form (with $\ell_{m'+1} = \ell'$):
	\begin{eqnarray*}
		\aheap\ =\ \myset{(\ell,\ell',\ell',\ell',\ell',\ell',\ell_1)}
		 & \cup 
		 & \setof{(\ell_i,\ellg{p}{i},\ellg{q}{i},\ellg{y}{i},\ellg{z}{i},\ellg{u}{i},\ell_{i+1})}{i = 1, \ldots, m'}\\
		 & \cup & \setof{(\ellg{y}{i},\ell',\ell',\ell',\ell',\ell',\ell')}{i = 1,\ldots, m'}\\
		 & \cup & \setof{(\ellg{z}{i},\ell',\ell',\ell',\ell',\ell',\ell')}{i = 1,\ldots, m'}\\
		 & \cup & \setof{(\ellg{u}{i},\ell',\ell',\ell',\ell',\ell',\ell')}{i = 1,\ldots, m'}.
	\end{eqnarray*}
Still by  construction, we have $p_1 = q_1 \in \startp$, $p_{m'} \in \stopp{u}$ and $q_{m'} \in \stopp{v}$, and  there exists $i$ such that $p_{m'}$ and $q_{m'}$ are of the form $(i,\len{u_i})$ and $(i,\len{v_i})$, %\mycomment[me]{removed ``$p = $'' and ``$q = ''$} 
respectively.
% \mycomment[np]{instead of ``the first components of $p_{m'}$ and $q_{m'}$ are identical.''}
%$\astore(x) = \ell$, $\astore(\nil) = \ell'$, and
%	for all $i = 1,\dots,m$: $\ell_i^1 = \astore(p_i), \ell_i^2 = \astore(q_i)$, for some $p_i,q_i \in \pairs$ such that $p_i \myrel q_i$, with $\forall i < m$, $p_i \nextp{u} p_{i+1}$, $q_i \nextp{v} q_{i+1}$, $p_1 = q_1 \in \startp$ and $p_m, q_m \in \stopp{u}$, where the first components of $p_m$ and $q_m$ are identical.
This entails that the words $u(p_1).\dots,u(p_{m'})$ and
$v(p_1).\dots,v(p_{m'})$
are of form 
$u_{i_1}.\dots.u_{i_k}$ 
and %\mycomment[me]{replaced $i'$ by $j$}\mycomment[np]{change propagated below}
%$v_{i_1'}.\dots.v_{i_{k'}'}$, 
$v_{j_1}.\dots.v_{j_{k'}}$, 
for some sequences
$(i_1,\dots,i_k)$ and $(j_1,\dots,j_{k'})$ of elements in $\interv{1}{n}$, {and both words are identical}.
However the sequences $(i_1,\dots,i_k)$
and
$(j_1,\dots,j_{k'})$ may be distinct (but we  have $i_1 = j_1$ and $i_k = j_{k'}$).
To check that the PCP has a solution, we must therefore verify that there exists a structure of the  form above such that 
$(i_1,\dots,i_k) = (j_1,\dots,j_{k'})$. To this purpose, we introduce predicates that are satisfied when this condition does not hold, i.e., such that either $k\neq k'$ or $i_l \neq j_{l}$ for some $l \in \interv{2}{k-1}$.
%\mycomment[np]{added:} 
This is done by using the additional locations $\ellg{y}{i}$, $\ellg{z}{i}$ and $\ellg{u}{i}$
to relate the indices $I_l = \len{u_{i_1}.\dots,u_{i_{l-1}}}+1$ %\mycomment[np]{added +1 and replaced $i'$ by $j$} 
and $J_l = \len{v_{j_1}.\dots,v_{j_{l-1}}}+1$, corresponding to the beginning of the words 
$u_{i_l}$ and $v_{j_l}$ respectively in  $u_{i_1}.\dots.u_{i_k}$ 
and 
$v_{j_1}.\dots.v_{j_{k'}}$.
 The predicates relate the locations of the form $\ellg{y}{i}$, $\ellg{z}{i}$ and $\ellg{u}{i}$ using the relation $\succs$. %\mycomment[me]{instead of relations $\succs$ and $\nsuccp$}, 
More precisely, they are associated with rules that guarantee that all
 the countermodels of the right-hand side of the sequent will satisfy the following properties:
\begin{compactenum}
	\item %\mycomment[np]{deleted (not correct, or we must be more precise and state which predicate ensures which property): The countermodels  are exactly the structures that ensure} 
	$k = k'$ 
	and for all $l \in \interv{1}{k}$, $(\ellg{y}{I_l},\ellg{z}{J_l}) \in \succs$
	and $(\ellg{u}{I_l},\ellg{z}{J_l}) \not\in \succs$.
%	\mycomment[me]{instead of $(\ellg{u}{I_l},\ellg{z}{J_l}) \in \nsuccp$}.
%\mycomment[np]{instead of: ``the elements $\ellg{y}{{i_l}}$ and $\ellg{z}{j_l}$ are related by $\succs$ for $1\leq l\leq k$; note that $i_l,j_l$ and indices in $\interv{1}{n}$ not in $\interv{1}{m}$''}
	\item %\mycomment[np]{deleted: The countermodels are exactly the structures that ensure} 
	$i_l = j_l$ for $1\leq l\leq k$.
\end{compactenum}
%\mycomment[np]{added:} 
Note that the hypothesis of the theorem ensures that 
the locations $\ellg{y}{i}$, %\mycomment[np]{added:} 
$\ellg{z}{i}$ and $\ellg{u}{i}$ %\mycomment[me]{instead of $\ellg{y}{i},\ellg{y}{i},\ellg{u}{i}$} 
can be chosen in such a way that 
there is a {\em unique} location $\ellg{z}{i}$ satisfying  $(\ellg{y}{I_l},\ellg{z}{i}) \in \succs \wedge (\ellg{u}{J_l}, \ellg{z}{i}) \notin \succs$, %\mycomment[me]{instead of $(\ell,\ellg{z}{J_l}) \in \nsuccp$} 
thus property $1$ above can be used to relate the indices $I_l$ and $J_l$, which, in turns, allows us to enforce property $2$.
Two predicates are used to guarantee that condition 1 holds for the countermodels. Predicate $\succpbFirstCell$ is satisfied by those structures for which the condition is not satisfied for $l = 1$, 
%\mycomment[np]{instead of $\ellg{y}{i_1}$ and $\ellg{z}{j_1}$ are not related by $\succs$}  
and $\succpbInnerCell$ is satisfied for those structures for which either $k\neq k'$ or there is %\mycomment[np]{modifs} 
an $l \in \interv{1}{k-1}$ such that the condition is satisfied at $l$, but not at $l+1$. %\ellg{y}{i_l}$ and $\ellg{z}{j_l}$ are related by $\succs$ but $\ellg{y}{i_{l + 1}}$ and $\ellg{z}{j_{l+1}}$ are not. 
Thus the structures that satisfy $\lsol(x,\vec{v})$ and that are countermodels of the disjunction of $\succpbFirstCell$ and $\succpbInnerCell$ are exactly the structures for which $k = k'$ and 
$(\ellg{y}{I_l},\ellg{z}{J_l}) \in \succs \wedge (\ellg{u}{I_l},\ellg{z}{J_l}) \notin \succs$ %\mycomment[me]{replaced $\nsuccp$} 
%for $\ellg{y}{i_l}$ and $\ellg{z}{j_l}$ are related by $\succs$ 
for 
	$l = 1, \ldots, k$. Predicate $\pbSeq$ is then used to guarantee that for all $1\leq l \leq k$, $p_{I_l} = q_{J_l}$: this predicate is satisfied by the structures for which there is an $l$ such that 
	$(\ellg{y}{I_l},\ellg{z}{J_l}) \in \succs \wedge (\ellg{u}{I_l},\ellg{z}{J_l}) \notin \succs$ %\mycomment[me]{replaced $\nsuccp$}
	holds 
	%\mycomment[np]{instead of $\ellg{y}{i_l}$ and $\ellg{z}{i_l}$ are related by $\succs$} 
	but $p_{I_l} \neq q_{J_l}$.
	We first give the rules for predicate $\succpbFirstCell$. For conciseness, we allow for disjunctions in the right-hand side of the rules (they can be eliminated by transformation into dnf).
%whose models are structures in which this condition does not hold: 
\[
\begin{tabular}{rclr}
	$\succpbFirstCell(x,\vec{v})$ & $\Leftarrow$ & $\exists x'. ~ x \mapsto (\nil,\nil,\nil,\nil,\nil,x') * \succpbFirstCell'(x',\vec{v})$ \\
	$\succpbFirstCell'(x,\vec{v})$ & $\Leftarrow$ & $\exists x',y,z,u. ~ x \mapsto (p,q,y,z,u,x') * \lsol_{p',q'}(x,\vec{v}) * P(y,\nil) $ \\
	& & $\qquad *\ P(z,\nil) * P(u,\nil) * (\nsuccp(y,z) \vee \succp(u,z))$,  for every $p,q,p',q'\in \pairs$ 
\end{tabular}
\]
Note that since $y,z,u$ are allocated in distinct predicates, they must be distinct, 
hence 
$\nsuccp(y,z)$ is equivalent to $\neg \succp(y,z)$ and
$\succp(u,z)$ is equivalent to $\neg \nsuccp(u,z)$.

	\newcommand{\PPP}{\aform'}

	We now define the rules for predicate $\succpbInnerCell$, which is meant to ensure that  
the condition ``$(\ellg{y}{I_l},\ellg{z}{J_l}) \in \succs$
and
$(\ellg{u}{I_l},\ellg{z}{J_l}) \not \in \succs$'' ($\dagger$) propagates, i.e., that if it holds for some $l$
then it also holds for $l+1$.
%if $(\ell_{x_i}^3,\ell_{y_j}^4) \in \succs$, with $i < k$ and $j < k'$, then
%$(\ell_{x_{i+1}}^3,\ell_{y_{j+1}}^4)  \in \succs$. 
This predicate has additional parameters $y,y',z,z',u,u'$ corresponding to the 
locations $\ellg{y}{I_l},\ellg{y}{I_{l+1}},\ellg{z}{J_j},\ellg{z}{J_{j+1}},\ellg{u}{I_l},\ellg{u}{I_{l+1}}$ 
which ``break'' the propagation of $(\dagger)$. By definition $y,y',z,z',u,u'$ must be chosen in such a way that the \tformula 
$\succp(y,z) * \nsuccp(u,z) * (\nsuccp(y',z') \vee \succp(u',z'))$ holds.
%The predicates $\succpbInnerCell$ allocates the location $\ell$ and calls a predicate $\succpbInnerCell_{0,0}$. 
The predicates $\succpbInnerCell_{a,b}$ with $a,b\in \{0,1,2\}$ allocate all the locations $\ell_1,\dots,\ell_{m'}$
and in particular the ``faulty'' locations  associated with $y,y',z,z',u,u'$.
Intuitively, 
$a$ (resp.\ $b$) denote the number of variables in $\{ y,y' \}$ 
(resp.\ $\{ z,z' \}$)
that have been allocated. 
The rules for predicates $\succpbInnerCell_{a,b}$ are meant to guarantee that the following conditions are satisfied for variables $y$ and $y'$ (similar constraints hold for $z$ and $z'$):
\begin{itemize}
	\item $y$ is allocated before $y'$,
	\item $y$ is allocated for a variable $p$ corresponding to the beginning of a word $(p\in \startp$),
	\item when $y$ has been allocated, no variable $p\in \startp$ can occur on the right-hand side of a points-to atom until $y'$ is allocated.
\end{itemize}
Several cases are distinguished depending on whether the locations associated with $y$ and $z$ (resp.\ $y'$ and $z'$) are in the same heap image of a location or not. 
%\mycomment[me]{whether the variables $y$ and $z$ (resp (...)) are mapped/refer to the same location or not?}\mycomment[np]{I think the current formulation is more precise. The pb is not related to the mapping between variables and locations (store) or whether we have $\astore(y) = \astore(z)$, but rather whether $\astore(y)$, $\astore(z)$ occur in the same heap tuple}
Note that $u$ and $u'$ are allocated in the same rules as $y$ and $y'$ respectively. 
%\mycomment[me]{end of sentence to remove if my addition is kept} and that $y$ is always allocated before $y'$ and $z$ before $z'$.
The predicate $\succpbInnerCell$ also tackles the case where $k \not = k'$. This corresponds to the case where 
the recursive calls end with $\succpbInnerCell_{1,2}$ or $\succpbInnerCell_{2,1}$ in the last rule below, 
meaning that  $(\dagger)$ holds for some $i$, with either $i = k$ and $i < k'$ or $i = k'$ and $i < k$.
For the sake of conciseness and readability, we denote by $\vec{w}$ the vector of variables $\vec{v},y,y',z,z',u,u'$ in the rules below. We also denote by $\PPP(y,z,u)$ the formula $P(y,\nil) * P(z,\nil) * P(u,\nil)$.
	
	%\mycomment[np]{non trivial corrections below (addition of the cond ($a \not = 1$ or $p \not \in \startp$) or ($b \not = 1$ or $q \not \in  \startp$) in some of rules}\mycomment[me]{I am spending too much time reconstructing this, this will go faster if we take a couple of minutes to discuss}
	
	{\small
		\[
		\begin{tabular}{rclr}
			$\succpbInnerCell(x,\vec{w})$ & $\Leftarrow$ & $x \mapsto (\nil,\nil,\nil,\nil,\nil,x') * \succpbInnerCell_{0,0}(x',\vec{w})$ \\
			& & $\qquad * \succp(y,z) * \nsuccp(u,z) * (\nsuccp(y',z') \vee \succp(u',z'))$ \\
			
			$\succpbInnerCell_{a,b}(x,\vec{w})$ & $\Leftarrow$ & $\exists x',y'',z'',u''. x \mapsto (p,q,y'',z'',u'',x') * \succpbInnerCell_{a,b}(x',\vec{w}) * \PPP(y'',z'',u'') $ \\
			&&  \quad 
			if ($a \not = 1$ or $p \not \in \startp$) and  ($b \not = 1$ or $q \not \in  \startp$) \\

			$\succpbInnerCell_{0,0}(x,\vec{w})$ & $\Leftarrow$ & $\exists x'. x \mapsto (p,q,y,z,u,x') * \succpbInnerCell_{1,1}(x',\vec{w}) * \PPP(y,z,u)$ \\
			& & \quad if $p,q \in \startp$ \\
			
			$\succpbInnerCell_{0,1}(x,\vec{w})$ & $\Leftarrow$ & $\exists x'. x \mapsto (p,q,y,z',u,x') * \succpbInnerCell_{1,2}(x',\vec{w}) * \PPP(y,z',u)$ \\
			& & \quad if $p,q \in \startp$ \\
			
			$\succpbInnerCell_{1,0}(x,\vec{w})$ & $\Leftarrow$ & $\exists x'. x \mapsto (p,q,y',z,u',x') * \succpbInnerCell_{2,1}(x',\vec{w}) * \PPP(y',z,u')$ \\
			& & \quad if $p,q \in \startp$ \\
			
			$\succpbInnerCell_{1,1}(x,\vec{w})$ & $\Leftarrow$ & $\exists x'. x \mapsto (p,q,y',z',u',x') * \succpbInnerCell_{2,2}(x',\vec{w}) * \PPP(y',z',u')$ \\
			& & \quad if $p,q \in \startp$ \\

			%$\succpbInnerCell_{a,b}(x,\vec{w})$ & $\Leftarrow$ & $x \mapsto (p,q,y',z',\nil) * \PPP(y',z',u')$ \\
			%& & \quad if $(a,b) \in \{ (2,2), (2,1), (1,2) \}$ \\
			$\succpbInnerCell_{0,b}(x,\vec{w})$ & $\Leftarrow$ & $\exists x',z''. x \mapsto (p,q,y,z'',u,x') * \succpbInnerCell_{1,b}(x',\vec{w}) * \PPP(y,z'',u)$ \\
			& & \quad if $p \in \startp$ and  ($b \not = 1$ or $q \not \in  \startp$) \\
			
			$\succpbInnerCell_{1,b}(x,\vec{w})$ & $\Leftarrow$ & $\exists x',z''. x \mapsto (p,q,y',z'',u',x') * \succpbInnerCell_{2,b}(x',\vec{w}) * \PPP(y',z'',u')$ \\
			& & \quad if $p \in \startp$ and  ($b \not = 1$ or $q \not \in  \startp$)  \\
			
			%$\succpbInnerCell_{1,2}(x,\vec{w})$ & $\Leftarrow$ & $\exists z''. x \mapsto (p,q,y',z'',\nil) * P(y',\nil) * P(z'',\nil)$ \\
			%& & \quad if $q \in \startp$ \\
			
			$\succpbInnerCell_{a,0}(x,\vec{w})$ & $\Leftarrow$ & $\exists x',y'',u''. x \mapsto (p,q,y'',z,u'',x') * \succpbInnerCell_{a,1}(x',\vec{w}) * \PPP(y'',z,u'')$ \\
			& & \quad if $q \in \startp$ and ($a \not = 1$ or $p \not \in \startp$)\\
			
			$\succpbInnerCell_{a,1}(x,\vec{w})$ & $\Leftarrow$ & $\exists x',y'',u''. x \mapsto (p,q,y'',z',u'',x') * \succpbInnerCell_{a,2}(x',\vec{w}) * \PPP(y'',z',u'')$ \\
			& & \quad if $q \in \startp$ and ($a \not = 1$ or $p \not \in \startp$) \\
			
			%$\succpbInnerCell_{a,2}(x,\vec{w})$ & $\Leftarrow$ & $\exists z''. x \mapsto (p,q,y',z'',\nil)  * P(y'',\nil) * P(z',\nil)$ \\
			%& & \quad if $p \in \startp$ \\
			
			$\succpbInnerCell_{a,b}(x,\vec{w})$ & $\Leftarrow$ & $\exists y'',z'',u''. x \mapsto (p,q,y'',z'',u'',\nil) * \PPP(y'',z'',u'')$ \\
			& & if  $(a,b) \in \{ (2,2), (2,1), (1,2) \}$
		\end{tabular}
		\]
	}
	%\mycomment[np]{$q$ instead of $p$ in rule $-2$ above}
	%By definition of the rules above,
 A straightforward induction %on $i$ 
 permits to verify that 
	if the considered structure does not satisfy the formula
	$\succpbFirstCell(x,\vec{v}) \vee \exists y,z,y',z',u,u'.\succpbInnerCell(x,\vec{w})$
	then necessarily $k = k'$ and for all
	$l \in \interv{1}{k}$, we have $(\ellg{y}{I_l},\ellg{z}{J_l}) \in \succs \wedge (\ellg{u}{I_l},\ellg{z}{J_l}) \not \in \succs$.
	
%	There remains to check that $p_{x_i} = q_{y_i}$, for all $i \in \interv{1}{k}$, using the  predicate $\pbSeq(x,\vec{v})$ allocating a structure not satisfying this condition (assuming 
%	the condition $(\dagger)$ above is fulfilled).
%	This predicate allocates the location $\ell$ and introduces existential
%	variables $y,y',z$ denoting the faulty locations $\ell_{x_i}, \ell_{y_i}$ and 
%	$\ell_{x_{i+1}}$, i.e., the locations corresponding to the index $i$ such that 
%	$p_{x_i} \not = q_{y_i}$. By $(\dagger)$, these locations must be chosen in such a way that the constraints
%	$\succp(y,z)$ and $\nsuccp(y',z)$ are satisfied. The predicate $\pbSeq(x,\vec{v})$ also guesses pairs $p,q$ such that $p \not = q$ (denoting the distinct pairs 
%	$p_{x_i}$ and $q_{y_i}$) and calls the predicate $\pbSeq_{p,q}^{(0,0)}$ to allocate all the remaining locations.
%	As for the previous rules, the predicates $\pbSeq_{p,q}^{a,b}$, for $p,q \in \startp$
%	$a,b\in \{0,1 \}$  allocate $\ell_1,\dots,\ell_m$, where $a$ (resp.\ $b$)
%	denotes the number of variables in $\{ y \}$ (resp. $\{ z\}$) that have already been allocated.
%	In the rules below, we denote by $\vec{u}$ the vector $\vec{v},y,z,u$.
%	In all the rules we have $p',q'\in \pairs$. % and $p''\in \startp$. 
%	\mycomment[np]{condition $p,q\in \startp$ added in the first rule below}
There remains to check that $p_{I_i} = q_{J_i}$  %\mycomment[np]{not $p_{x_i} = q_{y_i}$,} 
for all $i \in \interv{1}{k}$. To this aim, we design  an atom  $\pbSeq(x,\vec{v})$ that will be satisfied by structures not validating this condition, assuming 
the condition $(\dagger)$ above is fulfilled. 
% (which is guaranteed by the fact that the considered structure falsifies $\succpbFirstCell(x,\vec{v}) \vee \exists y,z,y',z',u,u'.\succpbInnerCell(x,\vec{w})$).
This predicate allocates the location $\ell$ and introduces existential
variables $y,z,u$ 
%\mycomment[np]{instead of $y,y',z$} 
denoting the faulty locations $\ellg{y}{I_i}, \ellg{z}{J_i}$ and 
$\ellg{u}{I_i}$, %\mycomment[np]{instead of $\ellg{p}{i}, \ellg{q}{i}$ and 
%$\ellg{p}{i+1}$,} 
i.e., the locations corresponding to the index $i$ such that 
$p_{I_i} \not = q_{J_i}$. 
By $(\dagger)$, these locations must be chosen in such a way that the constraints
$\succp(y,z)$ and $\nsuccp(y,z)$ are satisfied. The predicate $\pbSeq(x,\vec{v})$ also guesses pairs $p,q$ such that $p \not = q$ (denoting the distinct pairs 
$p_{x_i}$ and $q_{y_i}$) and invokes the predicate $\pbSeq_{p,q}^{(0,0)}$ to allocate all the remaining locations.
As for the previous rules, the predicates $\pbSeq_{p,q}^{a,b}$, for $p,q \in \startp$
$a,b\in \{0,1 \}$  allocate $\ell_1,\dots,\ell_{m'}$, where $a$ (resp.\ $b$)
denotes the number of variables in $\{ y \}$ (resp. $\{ z\}$) that have already been allocated.
In the rules below, we denote by $\vec{u}$ the vector $\vec{v},y,z,u$.
In all the rules we have $p',q'\in \pairs$. % and $p''\in \startp$. 
%\mycomment[np]{condition $p,q\in \startp$ added in the first rule below}
	
	%\mycomment[np]{some typos fixed below ($\vec{w} \leftarrow \vec{u}$)}
	
	{\small
		\[
		\begin{tabular}{rclr}
			$\pbSeq(x,\vec{v})$ & $\Leftarrow$ & $\exists y,z,u. ~ x \mapsto (\nil,\nil,\nil,\nil,\nil,x') * \pbSeq_{p,q}^{0,0}(x',\vec{u}) * \succp(y,z) * \nsuccp(u,z)$  \\ 
			& & \quad if $p \not = q$ and $p,q \in \startp$ \\
			$\pbSeq_{p,q}^{a,b}(x,\vec{u})$ & $\Leftarrow$ & $\exists x',y'',z'',u''. x \mapsto (p',q',y'',z'',u'',x') * \pbSeq_{p,q}^{a,b}(x,\vec{u})$ \\  
			& & $\qquad *\ \PPP(y'',z'',u'')$ \\
			$\pbSeq_{p,q}^{0,0}(x,\vec{u})$ & $\Leftarrow$ & $\exists x'. x \mapsto (p,q,y,z,u,x') * \pbSeq_{p,q}^{1,1}(x,\vec{u}) * \PPP(y,z,u)$  \\
			$\pbSeq_{p,q}^{0,b}(x,\vec{u})$ & $\Leftarrow$  & $\exists x',z''. x \mapsto (p,q',y,z'',u,x') * \pbSeq_{p,q}^{1,b}(x,\vec{u}) * \PPP(y,z'',u)$ \\
			$\pbSeq_{p,q}^{a,0}(x,\vec{u})$ & $\Leftarrow$ & $\exists x',y'',u''. x \mapsto (p',q,y'',z,u'',x') * \pbSeq_{p,q}^{a,1}(x,\vec{u}) * \PPP(y'',z,u'')$ \\
			$\pbSeq_{p,q}^{1,1}(x,\vec{u})$ & $\Leftarrow$  & $\exists y'',z'',u''. x \mapsto (p',q',y'',z'',u'',\nil) * \PPP(y'',z'',u'')$ \\
		\end{tabular}
		\]
	}
	
	The PCP has a solution iff %\mycomment[me]{and only if?}\mycomment[np]{yes (converse below)} 
	the sequent 
	\[\lsol(x,\vec{v}) \vdashr \succpbFirstCell(x,\vec{v}), \exists y,z,y',z',u,u'.\succpbInnerCell(x,\vec{w}), \pbSeq(x,\vec{u})\]
	has a countermodel.
	Indeed, if a structure satisfying the atom $\lsol(x,\vec{v})$ but not the disjunction $\succpbFirstCell(x,\vec{v}) \vee \exists y,z,y',z',u,u'.\succpbInnerCell(x,\vec{w}) \vee \pbSeq(x,\vec{u})$
	exists, then as explained above, there exists a word %\mycomment[np]{replaced $i'$ by $j$}
	$u_{i_1}.\dots.u_{i_k} = v_{j_1}.\dots.v_{j_{k'}}$, 
	with $(i_1,\dots,i_k) = (j_1,\dots,j_{k'})$.
	Conversely, if a solution of the PCP exists, then by using the locations $\aloc_l,\aloc_l',\aloc_l''$ in the hypothesis of the lemma as $\ellg{y}{I_l},\ellg{z}{J_l},\ellg{u}{I_l}$ it is easy 
	to construct a 
	a structure satisfying $\lsol(x,\vec{v})$. 
	%Further, this structure does not satisfy $\succpbFirstCell(x,,\vec{v}) \vee \exists y,z,y',z',u,u'.\succpbInnerCell(x,\vec{w}) \vee \pbSeq(x,\vec{u})$.
	Further, by hypothesis, since $(\aloc_l,\aloc_l') \in \succs$ and $(\aloc_l'',\aloc_l') \not\in \succs$, 
	we have $(\ellg{y}{I_l},\ellg{z}{J_l}) \in \succs$ and $(\ellg{u}{I_l},\ellg{z}{J_l}) \notin \succs$ %\mycomment[me]{instead of $\in \succs$} 
	for all $l = 1,\dots,k$. 
	Thus  $\succpbFirstCell(x,\vec{v})$ and $\exists y,z,y',z',u,u'.\succpbInnerCell(x,\vec{w})$ do not hold. 
	To fulfill $\neg \pbSeq(x,\vec{u})$ we have to ensure that, for all $i,j\in \interv{1}{k}$, 
	we have $(\ellg{y}{I_i},\ellg{z}{J_j}) \in \succs \wedge  (\ellg{u}{I_i},\ellg{z}{J_j}) \not \in \succs \implies
	p_{I_i} = q_{J_i}$. 
Since the considered word is a solution of the PCP, we have $p_{I_i} = q_{J_i}$ %\mycomment[np]{not $p_i = q_i$} 
	for all $i = 1,\dots,k$,  hence
	$\neg \pbSeq(x,\vec{u})$ is satisfied.
   \qed
\end{proof}

\section{Eliminating Equations and Disequations}

In this section, 
we show that the equations and disequations can always be eliminated from established sequents (in exponential time), while preserving equivalence. 
The intuition is that the equations can be discarded by instantiating the inductive rules,
while the disequations can be replaced by assertions that the considered variables 
are allocated in disjoint parts of the heap.

%\section{Allocated Variables and Roots}
 
 %\mycomment[np]{more bla bla}
 
We first introduce an additional restriction 
on \pcSIDs that is meant to ensure that the set of free variables 
allocated by a predicate atom is the same in every unfolding.
The \pcSID satisfying this condition are called {\em \alloccompatible}. 
We will show that every \pcSID can be reduced to 
an equivalent \alloccompatible set.
Let $\allocf$ be a function mapping each predicate symbol $p$  to a 
subset  of $\interv{1}{\ar{p}}$.
For any disjunction-free  formula $\aform$, we denote by $\alloc{\aform}$  the set 
of variables $x\in \fv{\aform}$ such that $\aform$ contains an atom of the form 
$x \mapsto (y_1,\dots,y_\rank)$ or
$p(z_1,\dots,z_n)$, with $x = z_i$ for some $i \in \alloc{p}$.
%inductively defined as follows:
%{\small 
%\[
%\begin{tabular}{llll}
%%$\alloc{\emp}$ & $=$ & $\emptyset$ \\
%$\alloc{\atform}$ & $=$ & $\emptyset$ & if $\atform$ is a \tformula (or $\emp$)\\
%$\alloc{x \mapsto (y_1,\dots,y_\rank)}$ & $=$ & $\{ x \}$ \\
%$\alloc{p(x_1,\dots,x_n)}$ & $=$ & $\{ x_i \mid i \in \alloc{p} \}$ & if $p \in \preds$  \\
%$\alloc{\aform_1 * \aform_2}$ & $=$ & $ \alloc{\aform_1} \cup \alloc{\aform_2}$ \\
%$\alloc{\exists x. ~\aform}$ & $=$ & $ \alloc{\aform} \setminus \{ x \}$  \\
%\end{tabular}
%\]
%}
%
\newcommand{\afun}{{\frak f}}

 \newcommand{\expl}[1]{#1^*}

 \begin{definition}
 An established \pcSID $\asid$ is {\em \alloccompatible} if for all rules $\anatom \Leftarrow \aform$ in $\asid$, we have $\alloc{\anatom} = \alloc{\aform}$ .
A sequent $\aform \vdashr \aseq$ is {\em \alloccompatible} if $\asid$ is \alloccompatible.
 \end{definition}

%\mycomment[np]{added:} 
%Intuitively, $\alloc{\aform}$ is meant to contain  the free variables of $\aform$ that are allocated in  %\mycomment[me]{all models?} 
%the models of $\aform$.
% The fact that $\asid$ is \alloccompatible ensures that this set does not depend on the considered model of $\aform$.   %\mycomment[np]{ex:}
% 

% In the remainder of the paper, we assume that all the considered sequents are \alloccompatible. This is justified by the following:
 \begin{lemma}
 \label{lem:alloccomp}
 There exists an algorithm which, for every sequent $\aform \vdashr \aseq$,  
 computes an equivalent 
 \alloccompatible  sequent $\aform' \vdashsid{\asid'} \aseq'$. Moreover, this algorithm runs in exponential time  
and $\widt{\aform' \vdashsid{\asid'} \aseq'} = \bigO(\widt{\aform \vdashr \aseq}^2)$.
  \end{lemma}
  \begin{proof}
%  \optionalProof{Lemma \ref{lem:alloccomp}}{
 We associate all pairs $(p,A)$ where $p \in \preds$ and $A \subseteq \interv{1}{\ar{p}}$ with fresh, pairwise distinct predicate symbols $p_A \in \preds$, with the same arity as $p$, and we set $\alloc{p_A} = A$.
 For each disjunction-free  formula $\aform$, we denote by $\expl{\aform}$ the set of formulas obtained from $\aform$ by replacing every predicate atom
 $p(\vec{x})$ by an atom $p_A(\vec{x})$ with $A \subseteq \interv{1}{\ar{p}}$.
 Let $\asid'$ be the set of \alloccompatible rules of the form 
 $p_A(\vec{x}) \Leftarrow \aformB$, where $p(\vec{x}) \Leftarrow \aform$ is a rule in $\asid$
and $\aformB \in \expl{\aform}$. Note that the symbols 
$p_A$ may be encoded by words of length $\bigO(\len{p} + \ar{p})$, thus for every $\aformB \in \expl{\aform}$ we have $\widt{\aformB} = \bigO(\widt{\aform}^2)$, hence $\widt{\asid'} = \bigO(\widt{\asid}^2)$. 
%\mycomment[me]{to discuss together}\mycomment[np]{is this ok?}
%\mycomment[me]{If I understand correctly, if $\aform$ is of the form $\alpha_1\circ\cdots \circ \alpha_n$, $k$ denotes the maximal size of a predicate in $\aform$ and $m$ its maximal arity, then $\widt{\aform} \leq n\cdot k\cdot m$ and $\widt{\aformB} \leq n\cdot (k+m)\cdot m \leq (n\cdot k \cdot m)^2$? But is the maximal size of a predicate in $\aformB$ still $k$?}\mycomment[np]{We must create one predicate 
%for each pair $(p,A)$, where $p \in \preds$ and $A \subseteq \{ 1,\dots,\ar{p}\}$.
%These predicates can be named for instance  $p.a$, where $a$ is a binary number $a_1.\dots,a_{\ar{p}}$ of length $\ar{p}$, where $a_i = 1$ iff $a_i \in A$. Since the maximal length of the predicates is $k+m$}\mycomment[me]{I know we have already discussed this issue but I no longer remember the argument: how do we guarantee that there is no naming conflict with symbols of size $k+m$, for example that $p.a$ is not a name that already occurs?}
%\mycomment[np]{to discuss. I do not think that this is a problem if there is a conflict with an existing symbol, i.e., if we have $q = p.a$ for some $q$ in the initial pb, because the new symbols and the initial ones do not interfere. The important point is that $p.q  = q.b \implies p=q \wedge a=b$.}
We show by induction on the satisfiability relation that the following equivalence holds for every structure $(\astore,\aheap)$:
$(\astore,\aheap) \modelsr \aform$ iff there exists $\aformB \in \expl{\aform}$ such that
$(\astore,\aheap) \modelssid{\asid'} \aformB$. For the direct implication, we also prove that $\alloc{\aformB} = \{ x \in \fv{\aform} \mid \astore(x) \in \dom{\aheap}\}$.
\begin{compactitem}
\item{The proof is immediate if $\aform$ is a \tformula, since $\expl{\aform} = \{ \aform \}$, and the truth value of $\aform$ does not depend on the considered \pcSID. Also, by definition $\alloc{\aform} = \emptyset$ and all the models of $\aform$ have empty heaps.}
\item{If $\aform$ is of the form $x \mapsto (y_1,\dots,y_n)$, then $\expl{\aform} = \{ \aform \}$ and the truth value of $\aform$ does not depend on the considered \pcSID. Also, $\alloc{\aform} = \{ x \}$ and
we have $\dom{\aheap} = \{ \astore(x) \}$ for every model $(\astore,\aheap)$ of $\aform$.}
\item{Assume that $\aform = p(x_1,\dots,x_{\ar{p}})$. If $(\astore,\aheap) \modelsr \aform$ then
there exists a formula $\aformC$ such that $\aform \unfoldto{\asid} \aformC$ and
$(\astore,\aheap) \modelsr \aformC$. 
By the induction hypothesis, %we deduce that
there exists $\aformB \in \expl{\aformC}$ such that 
$(\astore,\aheap) \modelssid{\asid'} \aformB$ and $\alloc{\aformB} = \{ x\in \fv{\aformC}  \mid  \astore(x) \in \dom{\aheap} \}$. Let $A = \{ i \in \interv{1}{\ar{p}} \mid \astore(x_i) \in \dom{\aheap} \}$, so that $\alloc{\aformB} =\{ x_i \mid i \in A \}$.
%with. Note that we have in particular $1 \in A$, by Proposition \ref{prop:roots}.
By construction 
%By definition, $\alloc{\aformB} = \{ x_i \mid i \in A\}$, thus 
$p_A(x_1,\dots,x_n) \Leftarrow \aformB$ is \alloccompatible, and 
therefore $p_A(x_1,\dots,x_n) \unfoldto{\asid'} \aformB$, which entails that $(\astore,\aheap) \modelssid{\asid'} p_A(x_1,\dots,x_n)$.
By definition of $A$, $\alloc{p_A(x_1,\dots,x_n)} = \{ x\in \fv{\aform}  \mid  \astore(x) \in \dom{\aheap} \}$.

Conversely, assume that $(\astore,\aheap) \modelsr \aformB$ for some $\aformB \in \expl{\aform}$. 
Necessarily $\aformB$ is of the form $p_A(x_1,\dots,x_n)$ with $A \subseteq \interv{1}{\ar{p}}$.
We have $p_A(x_1,\dots,x_n) \unfoldto{\asid'} \aformB'$ and $(\astore,\aheap) \modelsr \aformB'$ for some formula $\aformB'$.
By definition of $\asid'$, we deduce that $p(x_1,\dots,x_n) \unfoldto{\asid} \aformC$, for some $\aformC$ such that $\aformB\in \expl{\aformC}$. 
 By the induction hypothesis, 
$(\astore,\aheap) \modelsr \aformC$, thus $(\astore,\aheap) \modelsr p(x_1,\dots,x_{\ar{p}})$. Since $p(x_1,\dots,x_{\ar{p}}) = \aform$, we have the result.
}
\item{Assume that $\aform = \aform_1 * \aform_2$.
If $(\astore,\aheap) \modelsr \aform$ then there exist disjoint heaps $\aheap_1,\aheap_2$ such that $(\astore,\aheap_i) \modelsr \aform_i$, for all $i = 1,2$ and $\aheap = \aheap_1 \dunion \aheap_2$. By the induction hypothesis, this entails that there exist formulas $\aformB_i \in \expl{\aform_i}$ for $i = 1,2$ such that 
$(\astore,\aheap_i) \modelssid{\asid'} \aformB_i$ and  $\alloc{\aformB_i} = \{ x \in \fv{\aform_i} \mid \astore(x) \in \dom{\aheap_i} \}$.  
Let $\aformB = \aformB_1 * \aformB_2$.
It is clear that  $(\astore,\aheap) \modelssid{\asid'} \aformB_1 * \aformB_2$ and  $\alloc{\aformB} = \alloc{\aformB_1 * \aformB_2} = \alloc{\aformB_1} \cup \alloc{\aformB_2} = \{ x \in \fv{\aform_1}\cup \fv{\aform_2} \mid \astore(x) \in \dom{\aheap} \} =
\{ x \in \fv{\aform} \mid \astore(x) \in \dom{\aheap} \}$. Since $\aformB_1 * \aformB_2 \in \expl{\aform}$, we obtain the result.

Conversely, assume that  there exists $\aformB \in \expl{\aform}$ such that
$(\astore,\aheap) \modelssid{\asid'} \aformB$. % and $\alloc{\aformB} = \{ x \in \fv{\aform} \mid \astore(x) \in \dom{\aheap \}$.
Then $\aformB = \aformB_1 * \aformB_2$ with $\aformB_i \in \expl{\aform_i}$, and we have
$(\astore,\aheap_i) \modelssid{\asid'} \aformB_i$, for $i = 1,2$ with $\aheap = \aheap_1 \dunion \aheap_2$.
%We show that $\alloc{\aformB_i} = \{ x \in \fv{\aform_i} \mid \astore(x) \in \dom{\aheap_i}\}$. Let $x\in \alloc{\aformB_i}$. By definition we have $x\in \fv{\aformB_i} \subseteq \fv{\aform_i}$ and by Lemma \ref{lem:alloc}, we have $\astore(x) \in \dom{\aheap_i}$. Now, assume that $x \in \fv{\aformB_i}$ and $\astore(x) \in \dom{\aheap_i}$. Then necessarily
%. $x \in \fv{\aformB}$ and $\astore(x) \in \dom{\aheap}$, thus 
%$x\in \alloc{\aformB} = \alloc{\aformB_1} \cup \alloc{\aformB_2}$. If $x \in \alloc{\aformB_j}$ for $j \not = i$ then 
%we have by Lemma \ref{lem:alloc}, we have $\astore(x) \in \dom{\aheap_j}$, which contradicts the disjointness of $\aheap_1$ and $\aheap_2$. Thus $x \in \alloc{\aformB_i}$.
% since $\aheap_1 \cap \aheap_2 = \emptyset$.
Using the induction hypothesis,  we get that $(\astore,\aheap_i) \modelsr \aform_i$, hence
 $(\astore,\aheap) \modelsr \aform$.}

\item{Assume that $\aform = \exists y. \aformC$.
If $(\astore,\aheap) \modelsr \aform$ 
then $(\astore',\aheap) \modelsr \aformC$, for some store $\astore'$ coinciding with $\astore$ on every variable distinct from $y$.
By the induction hypothesis, this entails that  
there exists $\aformB \in \expl{\aformC}$ such that
$(\astore',\aheap) \modelssid{\asid'} \aformB$ and  $\alloc{\aformB} = \{ x \in \fv{\aformC} \mid \astore'(x) \in \dom{\aheap}\}$. Then 
$(\astore,\aheap) \modelssid{\asid'} \exists y. \aformB$, and we have $\exists y.\aformB \in \expl{\aform}$.
Furthermore,  $\alloc{\exists y.\aformB} = \alloc{\aformB} \setminus \{ y\} = \{ x \in \fv{\aformC} \setminus \{ y \} \mid \astore'(x) \in \dom{\aheap}\} = \{ x \in \fv{\aform} \mid \astore(x) \in \dom{\aheap}\}$. 

Conversely, assume that 
$(\astore,\aheap) \modelsr \aformB$, with $\aformB\in \expl{\aform}$.
Then $\aformB$ is of the form $\exists y. \aformB'$, with $\aformB' \in \expl{\aformC}$, thus
there exists a store $\astore'$, coinciding with $\astore$ on all variables other than $y$ 
such that $(\astore',\aheap) \modelsr \aformB'$.
By the induction hypothesis, this entails that
$(\astore',\aheap) \modelsr \aformB$, thus 
$(\astore,\aheap) \modelsr \exists y. \aformC$. Since $\exists y. \aformC = \aform$, we have the result.
} 
\end{compactitem}
 
Let $\aform',\aseq'$ be the sequence of formulas obtained from $\aform,\aseq$ by replacing every atom $\anatom$ by the disjunction of all the formulas in $\expl{\anatom}$.
It is clear that $\widt{\aform' \vdash_{\asid'} \aseq'} \leq \widt{\aform \vdash_{\asid} \aseq}^2$.
By the previous result, $\aform' \vdashsid{\asid'} \aseq'$ is equivalent to $\aform \vdashsid{\asid} \aseq$, hence $\aform' \vdashsid{\asid'} \aseq'$  fulfills all the required properties. Also, since each predicate $p$ 
is associated with $2^{\ar{p}}$  predicates $p_A$, it is clear that $\aform' \vdashsid{\asid'} \aseq'$ can be computed in time $\bigO(2^{\size{\aform \vdashsid{\asid} \aseq}})$. 
% \mycomment[me]{instead of in time $2^{\size{\aform \vdashsid{\asid} \aseq}}$}.
\qed
\end{proof}

\newcommand{\purelyspatial}{\constrained{\emptyset}}
\newcommand{\constrained}[1]{$#1$-constrained\xspace}
\newcommand{\nespatial}{\constrained{\{ \not \iseq \}}}
\newcommand{\espatial}{\constrained{\{ \iseq, \not \iseq \}}}
\newcommand{\mroot}{main root\xspace}
\begin{definition}
Let $P \subseteq \tpreds$.
A formula $\aform$ is {\em \constrained{P}} if
for every formula $\aformB$ such that $\aform \unfoldto{\asid} \aformB$, and for every symbol  $p \in \tpreds$ occurring in $\aformB$, we have $p \in P$.
A sequent $\aform \vdashr \aseq$ is {\em \constrained{P}} if
all the formulas in $\aform,\aseq$ are \constrained{P}.
\end{definition}

%In particular, if $\aform$ is \purelyspatial, then the unfoldings of $\aform$ contain no symbol in $\tpreds$.
 
\begin{theorem}
\label{theo:elimeq}
Let $P \subseteq \tpreds$.
There exists an algorithm that transforms every  \constrained{P} established sequent $\aform \vdashr \aseq$
into an equivalent \constrained{(P \setminus \{ \iseq, \not \iseq \})} established sequent $\aform' \vdashsid{\asid'} \aseq'$.
This algorithm runs in  exponential time
and $\widt{\aform' \vdashsid{\asid'} \aseq'}$ is %\mycomment[np]{not useful to give a precise bound} 
polynomial w.r.t.\  $\widt{\aform \vdashr \aseq}$. % \bigO(\widt{\aform \vdashr \aseq}^2)$.
\end{theorem}
\begin{proof}
We consider a \constrained{P}  established sequent $\aform \vdashsid{\asid} \aseq$.
% \mycomment[me]{Addition, in case the example is helpful to understand each step.}\mycomment[np]{nice, but too long for short version :(}
 This sequent is transformed in several steps.
 
\noindent\textbf{Step 1.}
The first step  consists in transforming all the formulas in $\aform,\aseq$ into disjunctions of symbolic heaps. 
Then for every  symbolic heap $\aformC$ occurring in the obtained sequent, we add all the variables freely occurring
in $\aform$ or $\aseq$ as parameters of every predicate symbol occurring in unfoldings of $\aformC$  (their arities are updated accordingly, and these variables are passed as parameters to each recursive call of a predicate symbol).
We obtain an equivalent sequent $\aform_1 \vdashsid{\asid_1} \aseq_1$, and if 
%\mycomment[np]{modif here: we must take into account the size of the variables} 
$v = \card{\fv{\aform} \cup \fv{\aseq}}$
denotes the total number of free variables occurring in $\aform,\aseq$, then it is easy to check (since the size of each of these variables is bounded by $\widt{\aform \vdashsid{\asid} \aseq}$) that 
$\widt{\aform_1 \vdashsid{\asid_1} \aseq_1} \leq v\cdot \widt{\aform \vdashsid{\asid} \aseq}^2$.
By Definition \ref{def:sequent},
we have $v \leq \widt{\aform \vdashsid{\asid} \aseq}$,
thus $\widt{\aform_1 \vdashsid{\asid_1} \aseq_1} = \bigO(\widt{\aform \vdashsid{\asid} \aseq}^3)$.

% \mycomment[me]{I assume the argument is that the size of the fresh predicate symbols is greater by at most $v$ than the previous symbols.}\mycomment[np]{which fresh predicates? we replace atoms of size $s$ by atoms of size $s+v$}\mycomment[me]{In other words we are changing the arities of the existing predicate symbols?}\mycomment[np]{yes, this is the simplest solution.}
%, and 
%$\size{\aform_1 \vdashsid{\asid_1} \aseq_1} = 2^{\size{\aform \vdashsid{\asid} \aseq}}$.

\noindent\textbf{Step 2.}
%\mycomment[np]{deleted: The second step consists in eliminating all equations occurring in the sequent or rules, which is done as follows. This is done in steps 2-3} 
All the equations involving an existential variable can be eliminated in a straightforward way 
by replacing each formula of the form $\exists x. (x \iseq y * \aform)$ with 
$\repl{\aform}{x}{y}$.
We then replace every formula $\exists \vec{y}. \aform$ with free variables $x_1,\dots,x_n$ by the disjunction of all the formulas
of the form $\exists\vec{z}. \aform\sigma * \bigAnd_{z\in \vec{z}, z' \in \vec{z} \cup \{ x_1,\dots,x_n \}, z\not = z'} z \not \iseq z'$ , where $\sigma$ is a substitution such that $\dom{\sigma} \subseteq \vec{y}$, $\vec{z} = \vec{y}  \setminus \dom{\sigma}$ and $\img{\sigma} \subseteq \vec{y} \cup \{ x_1,\dots,x_n\}$. 
%\mycomment[me]{Instead of $\img{\sigma} = \vec{y} \cup \{ x_1,\dots,x_n\}$}.
Similarly we replace every rule $p(x_1,\dots,x_n)\Leftarrow \exists \vec{y}. \aform$ by the
the set of rules $p(x_1,\dots,x_n)\Leftarrow \exists\vec{z}. \aform\sigma * \bigAnd_{z\in \vec{z}, z' \in \vec{z} \cup \{ x_1,\dots,x_n \}, z\not = z'} z \not \iseq z'$, where $\sigma$ is any substitution satisfying the conditions above.  
%\mycomment[np]{added:} 
Intuitively, this transformation ensures that all existential variables are associated to pairwise distinct locations, also distinct from any location associated to a free variable. The application of the substitution $\sigma$ captures all the rule instances for which this condition does not hold, by mapping  all variables that are associated with the same location to a unique representative.  
 We denote by $\aform_2 \vdashsid{\asid_2} \aseq_2$ the sequent thus obtained.
Let $v'$ be the maximal number of existential variables occurring in a rule in $\asid$. 
We have $v' \leq \widt{\aform \vdashsid{\asid} \aseq}$ (since the transformation in Step $1$ adds no existential variable).
Since at most one disequation is added for every pair of variables, and since the size of every variable is bounded by $\widt{\aform \vdashsid{\asid} \aseq}$, it is clear that $\widt{\aform_2 \vdashsid{\asid_2} \aseq_2} = \widt{\aform_1 \vdashsid{\asid_1} \aseq_1} + v'\cdot (v+v') \cdot(1+2*\widt{\aform \vdashsid{\asid}   \aseq})
= \bigO(\widt{\aform \vdashsid{\asid} \aseq}^3)$.

\noindent\textbf{Step 3.}
We replace every atom $\anatom = p(x_1,\dots,x_n)$ occurring in $\aform_2, \aseq_2$ or $\asid_2$ with pairwise distinct  variables $x_{i_1},\dots,x_{i_m}$ (with $m \leq n$ and $i_1 = 1$), by an atom $p_{\anatom}(x_{i_1},\dots,x_{i_m})$, where $p_{\anatom}$ is a fresh predicate symbol, associated with rules of the form 
%$p_{\anatom}(x_{i_1},\dots,x_{i_m}) \Leftarrow \replall{\aformB}{y_i}{x_i}{i \in \interv{1}{n}}$, where 
$p_{\anatom}(y_{i_1},\dots,y_{i_m}) \Leftarrow \replall{\aformB}{y_i}{x_i}{i \in \interv{1}{n}}\theta$, where 
$p(y_1,\dots,y_n) \Leftarrow \aformB$ is a rule in $\asid$
and $\theta$ denotes the substitution $\replall{}{x_{i_k}}{y_{i_k}}{i \in \interv{1}{m}}$. 
%\mycomment[me]{Would it be clearer to take the mapping $\theta: x_{i_k} \mapsto y_{i_k}$ and the rule $p_{\anatom}(y_{i_1},\dots,y_{i_m}) \Leftarrow \replall{\aformB}{y_i}{x_i}{i \in \interv{1}{n}}\theta$? This would fit with the example below.} \mycomment[np]{I am not sure that it is clearer, but it is certainly equivalent :) changed}
%\mycomment[me]{I think it is also necessary to define $\alloc{p_\alpha}$ to guarantee that the resulting SID is alloc-compatible.}\mycomment[np]{I do not think that this 
%is needed, see below (other option would be to apply the previous transformation to get alloc-compatible 
%rules, but one would have to check that this does not affect the (dis)equalities). EDIT: after rethinking I think that we cannot get rid of the \alloccompatible condition. The previous version was wrong.}
By construction,  $p_{\anatom}(x_{i_1},\dots,x_{i_m})$ is equivalent to $\anatom$.  We denote by $\aform_3 \vdashsid{\asid_3} \aseq_3$ the resulting sequent.
It is clear that $\aform_3 \vdashsid{\asid_3} \aseq_3$ is equivalent to $\aform \vdashsid{\asid} \aseq$.

By a straightforward induction on the derivation, we can show that %The transformation at Steps 2 and 3 \mycomment[me]{step 3?}\mycomment[np]{both are needed} guarantees that
%Further, by definition the previous transformation, 
all atoms occurring in an unfolding of the formulas in the sequent $\aform_3 \vdashsid{\asid_3} \aseq_3$ are of the form $q(y_1,\dots,y_{\ar{q}})$, where $y_1,\dots,y_{\ar{q}}$ are pairwise distinct, and that the considered unfolding also contains the disequation $y_i \not \iseq y_j$, for all $i \not = j$ 
%\mycomment[np]{added:} 
such that either $y_i$ or $y_j$ is an existential variable (note that if  $y_i$ and $y_j$ are both free then $y_i \not \iseq y_j$ is valid, since the considered stores are injective). 
%\mycomment[me]{can there be rules of the form $p(x) \Leftarrow \exists y. q(y,x)?$ In step 2, should we add $\bigAnd_{z\in \vec{z}, z' \in \myset{x_1, \dots, x_n}} z \not \iseq z'$?}\mycomment[np]{true}
%\mycomment[np]{some modifs} 
This entails that the rules that introduce a trivial equality $u \iseq v$ with $u \not = v$ %containing an equality $y_i \iseq y_j$ with $i \not = j$ 
are actually redundant, since unfolding any atom  $q(y_1,\dots,y_{\ar{q}})$ using such a rule yields a formula that is unsatisfiable. Consequently such rules can be eliminated without affecting the status of the sequent. 
%\mycomment[me]{can there be rules of the form $p(x,y) \Leftarrow q(y) * x\iseq y$? Maybe we could consider disjunction-free sequents and instantiate free variables accordingly}\mycomment[np]{if such a rule occurs then the disequation $x \not \iseq y$ must have been 
%introduced somewhere  thus the rule is useless. Or, if $x,y$ are free variables then since we consider injective models
%the equation is false}
All the remaining equations are of form $u \iseq u$ hence can be replaced by $\emp$.
%We thus obtain a sequent $\aform_3 \vdashsid{\asid_3} \aseq_3$ equivalent to $\aform \vdashsid{\asid} \aseq$ and containing no equality. 
We may thus assume that the sequent $\aform_3 \vdashsid{\asid_3} \aseq_3$ contains no equality. 
Note that by the above transformation all existential variables must be interpreted as pairwise distinct locations in any interpretation, and also be distinct from all free variables. %\mycomment[me]{Same comment on the addition of $\bigAnd_{z\in \vec{z}, z' \in \myset{x_1, \dots, x_n}} z \not \iseq z'$}
It is easy to see that the fresh predicates $p_{\anatom}$ may be encoded by words
of size at most $\widt{\aform \vdashsid{\asid} \aseq}$, thus
$\widt{\aform_3 \vdashsid{\asid_3} \aseq_3} \leq \widt{\aform \vdashsid{\asid} \aseq} \cdot\widt{\aform_2 \vdashsid{\asid_2} \aseq_2} = \bigO(\widt{\aform \vdashsid{\asid} \aseq}^4)$.
%\mycomment[np]{added:} 
By Lemma \ref{lem:alloccomp}, we may assume that
$\aform_3 \vdashsid{\asid_3} \aseq_3$ is \alloccompatible (note that the transformation given in the proof of Lemma \ref{lem:alloccomp} does not affect the disequations occurring in the rules).
% and 
%$\size{\aform_2 \vdashsid{\asid_2} \aseq_2} = \bigO(2^{v}.\size{\aform_1 \vdashsid{\asid_1} \aseq_1})$. %, thus
%$\size{\aform_2 \vdashsid{\asid_2} \aseq_2} = \bigO(2^{\size{\aform \vdashsid{\asid} \aseq}})$.

%\mycomment[np]{addition, and some changes in the proof}
\noindent\textbf{Step 4.}
We now ensure that all the  locations that are referred to are allocated.
Consider a symbolic heap $\aformC$ occurring in $\aform_3,\aseq_3$ and any $\asid_3$-model $(\astore,\aheap)$ of $\aformC$, where $\astore $ is injective.
For the establishment condition to hold, the only unallocated locations in $\aheap$ of $\aformC$ must correspond to locations $\astore(x)$ where $x$ is a free variable. 
%\mycomment[np]{modif, to avoid changing $\rank$ (not really important)}
We assume the sequent contains a free variable $u$ such that,
 for every tuple $(\ell_0,\dots,\ell_\rank)\in \aheap$, we have $\astore(u) = \ell_\rank$.
This does not entail any loss of generality, since 
we can always add a fresh variable $u$ to the considered problem:
after Step $1$, $u$ is passed as a parameter to all predicate symbols, and we may replace every points-to atom $z_0 \mapsto (z_1,\dots,z_\rank)$ occurring in $\aform_3$, $\aseq_3$ or $\asid_3$, by
$z_0 \mapsto (z_1,\dots,z_\rank,u)$ (note that this increases the value of $\rank$ by $1$).
It is clear that this ensures that $\aheap$ and $u$ satisfy the above property.
We also assume, w.l.o.g., that the sequent contains at least one variable $u'$ distinct from $u$. Note that, since $\astore$ is injective, the tuple $(\astore(u'),\dots,\astore(u'))$ cannot occur in $\aheap$, because its last component is distinct from $\astore(u)$.
We then denote by $\aform_4 \vdashsid{\asid_4} \aseq_4$ the sequent obtained from 
$\aform_3\vdashsid{\asid_3} \aseq_3$ by replacing every symbolic heap 
$\aformC$ in $\aform_3,\aseq_3$ by 
$\left(\bigAnd_{x \in (\fv{\aform_3} \cup \fv{\aseq_3}) \setminus \alloc{\aformC}} x \mapsto (u',\dots,u')\right) * \aformC$
%$\bigAnd_{x \in (\fv{\aform_3} \cup \fv{\aseq_3}) \setminus \alloc{\aformC}} (x \mapsto (u',\dots,u') * \aformC)$.
%\(\bigvee_{V \subseteq \fv{\aform_3} \cup \fv{\aseq_3}} (\bigAnd_{x \in V} x \mapsto (x,\dots,x) * \aformC)\). 
%\mycomment[me]{$V$ does not occur within the parenthesis?}\mycomment[np]{oops :) fixed} %
%\mycomment[me]{Actually I find these sentences misleading, maybe we directly move on to the sentence ``it is straightforward to check (...)''}\mycomment[np]{ok}
%It is clear that this transformation preserves the models of the considered symbolic heap, except that all the non-allocated locations $\ell$ are 
%mapped to a tuple $(\ell,\dots,\ell)$, which is necessarily different from the image of all the already allocated locations. 
It is straightforward to check that $(\astore,\aheap)\models \aformC$ iff there exists an extension $\aheap'$ of $\aheap$
such that 
$(\astore,\aheap') \models \left(\bigAnd_{x \in (\fv{\aform_3} \cup \fv{\aseq_3}) \setminus \alloc{\aformC}} x \mapsto (u',\dots,u')\right) * \aformC$, 
%$(\astore,\aheap') \models \bigAnd_{x \in (\fv{\aform_3} \cup \fv{\aseq_3}) \setminus \alloc{\aformC}} (x \mapsto (u',\dots,u') * \aformC)$, 
with $\locs{\aheap} = \locs{\aheap'} = \dom{\aheap'}$ and
 $\aheap'(\ell) = (\astore(u'),\dots,\astore(u'))$ for all $\ell \in \dom{\aheap'} \setminus \dom{\aheap}$. 
This entails that $\aform_4 \vdashsid{\asid_4} \aseq_4$ is valid if and only if	$\aform \vdashsid{\asid} \aseq$ is valid.
%\mycomment[me]{Is this an equivalence? Isn't it necessary to extend the heap in one direction?}\mycomment[np]{the sequents are equivalent, in the sense that they are both valid or both non valid,
%at the level of formulas indeed this is not a real equivalence: one need to restrict/extend the heap} 
%\mycomment[me]{Suggestion This entails that $\aform_4 \vdashsid{\asid_4} \aseq_4$ is valid if and only if	$\aform \vdashsid{\asid} \aseq$ is valid.}\mycomment[np]{changed, but note that this is exactly the def
%of equivalence for sequents (of course one could use a stronger form of equivalence which state that the counter-examples are the same, but this is not done)}

Consider a formula $\aformC$ in  $\aform_4,\aseq_4$  and some  unfolding $\aformC'$ of $\aformC$.
Thanks to the transformation in this step and the establishment condition, if $\aformC'$  contains a (free or existential) variable $x$ then it also contains an atom $x' \mapsto \vec{y}$ and 
a separating conjunction of equations $\atform$ such that $\atform \modelst x \approx x'$. 
Since all equations have been removed, $\atform = \emp$, thus 
$x = x'$.
Consequently, if $\aformC'$ contains a disequation $x_1 \not \approx x_2$ with $x_1\not = x_2$, then it also contains atoms $x_1 \mapsto \vec{y}_1$ and $x_2 \mapsto \vec{y}_2$. This entails that the disequation $x_1\not \iseq x_2$ is redundant, since it is a logical consequence of $x_1 \mapsto \vec{y}_1 * x_2 \mapsto \vec{y}_2$. We deduce that the satisfiability status of $\aform_4 \vdashsid{\asid_4} \aseq_4$ is preserved if all disequations are replaced by $\emp$.
\qed
\end{proof}

%\mycomment[me]{Addition of illustration for each step.}\mycomment[np]{some modifs, shifted to  appendix}

%\putInAppendix{
\begin{example}\label{ex:elimeq}
	We illustrate all of the steps in the proof above. %\mycomment[np]{some equalities added}
	\begin{description}
		\item[Step 1.] Consider the sequent $p(x_1, x_2) \vdashsid{\asid} r(x_1) * r(x_2)$, where $\asid$ is defined as follows: $\asid = \myset{r(x) \Leftarrow x \mapsto (x)}$. After Step 1 we obtain the sequent $p(x_1, x_2) \vdashsid{\asid_1} r'(x_1, x_2) * r'(x_2, x_1)$, where $\asid_1 = \myset{r'(x, y) \Leftarrow x \mapsto (x)}$.
		\item[Step 2.] This step transforms the formula $\exists y_1\exists y_2.\, p(x, y_1) * p(x, y_2)$ into the disjunction: 
		\[\begin{array}{rl}
			\exists y_1,y_2.\, p(x, y_1) * p(x, y_2) * y_1 \not\iseq y_2 * y_1 \not \iseq x * y_2 \not \iseq x & \vee\\
			\exists y_2.\, p(x, x) * p(x,y_2) * y_2 \not \iseq x & \vee\\
			\exists y_1.\, p(x, y_1) * p(x,x) * y_1 \not \iseq x& \vee\\
			p(x,x) * p(x, x)			
		\end{array}\]
		Similarly, the rule 
		$p(x) \leftarrow \exists z \exists u.~ x \mapsto (z) * q(z,u)$ is transformed into the set:
		\[
		\begin{array}{lll}
		p(x) & \leftarrow & x \mapsto (x) * q(x,x) \\
		p(x) & \leftarrow & \exists z.~ x \mapsto (z) * q(z,x) * z \not \iseq x  \\
		p(x) & \leftarrow & \exists u.~ x \mapsto (x) * q(x,u) * u \not \iseq x  \\
		p(x) & \leftarrow & \exists z \exists u.~ x \mapsto (z) * q(z,u) * z \not \iseq x * u \not \iseq x  * z \not \iseq u \\
\end{array}
\]
		\item[Step 3.] %\mycomment[np]{some modifs, please re-check} 
		Assume that $\asid$ contains the rules $p(y_1, y_2, y_3) \Leftarrow y_1\mapsto (y_2) * q(y_2, y_3) * y_1 \iseq y_3$ and $p(y_1, y_2, y_3) \Leftarrow y_1\mapsto (y_2) * r(y_2, y_3) * y_1 \iseq y_2$ and consider the sequent $p(x,y,x) \vdashsid{\asid} \emp$. Step 3 generates the sequent $p_\alpha(x,y) \vdashsid{\asid'} \emp$ (with $\alpha = p(x,y,x)$) where  
%		\[\asid'\ =\ \myset{p(y_1, y_2, y_3) \Leftarrow y_1\mapsto (y_1) * q(y_2, y_3),\, p_\alpha(y_1, y_2) \Leftarrow y_1\mapsto (y_1) * q(y_2, y_2)}.\]
		$\asid'$ contains the rules $p_\alpha(y_1, y_2) \Leftarrow y_1\mapsto (y_2) * q(y_2, y_1) * y_1 \iseq y_1$ and $p_\alpha(y_1, y_2) \Leftarrow y_1\mapsto (y_2) * r(y_2, y_1) * y_1 \iseq y_2$.
		The second rule is redundant, because $p_\alpha(y_1, y_2)$ is used only in a context where $y_1  \not \iseq y_2$ holds.

		%\mycomment[np]{initial rule deleted}
		\item[Step 4.] %\mycomment[np]{modifs} 
		Let $\aformC = p(x,y,z,z') * q(x,y,z,z') * z' \mapsto (z')$, assume $\alloc{\aformC} = \myset{x,z}$, and consider the sequent $\aformC \vdashsid{\asid} \emp$. Then  $\aformC$ is replaced by $p(x,y,z,z',u) * q(x,y,z,z',u) * z' \mapsto (z',u) * u \mapsto (x,x) * y \mapsto (x,x)$ (all non-allocated variables are associated with $(x,x)$, where $x$ plays the r\^ole of the variable $u'$ in Step $4$ above). Also, every points-to atom $z_0 \mapsto (z_1)$ in $\asid$ is replaced by $z_0 \mapsto (z_1,u)$. 
		
		%\mycomment[np]{adaptation to the new transformation + modifs}
	\end{description}
\end{example}
%There exists an algorithm transforming any \espatial sequent into an equivalent \nespatial sequent (see \cite{EIP21a}).

%Furthermore, this algorithm increases the width of the sequent only polynomially.

\section{Discussion}

%\mycomment[np]{to be reread}

The presented undecidability result is very tight.
Theorem \ref{theo:undec} applies to most  theories 
and the proof only uses very simple data structures (namely simply linked lists).
The proof of Theorem \ref{theo:undec} could be adapted 
(at the cost of  cluttering the presentation) to  handle quantifier-free entailments and 
even simpler 
inductive systems  
with at most one predicate atom on the right-hand side of each rule, in the spirit of word automata.

Our logic has only one sort of variables, denoting locations, thus one cannot directly describe structures in which the heap maps locations to  tuples containing both locations and data, ranging over disjoint domains. This is actually not restrictive: indeed, data can be easily 
encoded  in our framework by considering a  non-injective function $\mathtt{d}(x)$ mapping locations to data, and adding theory
predicates constructed on this function, such as $\mathtt{d}(x) \iseq \mathtt{d}(y)$ to state that two (possibly distinct) locations
$x,y$ are mapped to the same element. The obtained theory falls within the scope of Theorem \ref{theo:undec} (using $\mathtt{d}(x) \iseq \mathtt{d}(y)$ as the relation $\succp(x,y)$), provided the domain of the data is infinite. This shows that 
entailments with data disjoint from locations are undecidable, even if the  theory only contains  equations and disequations, except when the data domain is finite.

  \subsubsection*{Acknowledgments.}
  
  This work has been partially funded by the 
the French National Research Agency ({\tt ANR-21-CE48-0011}).
  The authors wish to thank Radu Iosif for his comments on an earlier version of the paper and for fruitful discussions. 
    
%\bibliographystyle{plainurl}
%\bibliography{SL}

\end{document}